\documentclass[twocolumn]{autart}    

\usepackage{array}
\usepackage[caption=false,font=normalsize,labelfont=sf,textfont=sf]{subfig}
\usepackage{textcomp}
\usepackage{stfloats}
\usepackage{url}
\usepackage{verbatim}
\usepackage{graphicx}
\usepackage{cite}
\usepackage{graphicx}
\usepackage{subfig}

\usepackage{cite}
\usepackage{amsmath,amssymb,amsfonts}
\usepackage{algorithm}
\usepackage{graphicx,epstopdf}
\usepackage{textcomp}
\usepackage{multicol}
\usepackage{cuted}
\usepackage{IEEEtrantools}
\usepackage{mathrsfs}

\usepackage{xcolor}

\usepackage[colorlinks=true, allcolors=blue]{hyperref}

\newtheorem{theorem}{Theorem}

\newtheorem{assumption}{Assumption}
\newtheorem{lemma}{Lemma}

\DeclareMathOperator{\diag}{diag}

\newenvironment{proof}[1][Proof]{\par\noindent\textbf{#1.} }{\hfill$\square$\par}

\begin{document}

\begin{frontmatter}

\title{The Soft-PVTOL: modeling and control\thanksref{footnoteinfo}} 

\thanks[footnoteinfo]{This paper was not presented at any conference. Corresponding author G.~Flores. Tel. +1 (956) 326-3297.}

\author[TAMIU]{Gerardo Flores}\ead{gerardo.flores@tamiu.edu}, 
\author[UTD]{Mark W. Spong}\ead{mspong@utdallas.edu}        

\address[TAMIU]{RAPTOR Lab, School of Engineering, College of Arts and Sciences, \\Texas A\&M International University, Laredo, TX 78041 USA}   
\address[UTD]{Erik Jonsson School of Engineering \& Computer Science, Department of Systems Engineering, \\University of Texas at Dallas, Richardson, TX 75080 USA}      


\begin{abstract}                          
This paper presents, for the first time, the soft planar vertical take-off and landing (Soft-PVTOL) aircraft. This concept captures the soft aerial vehicle's fundamental dynamics with a minimum number of states and inputs but retains the main features to consider when designing control laws. Unlike conventional PVTOL and multi-rotors, where altering position inevitably impacts orientation due to their underactuated design, the Soft-PVTOL offers the unique advantage of separating these dynamics, opening doors to unparalleled maneuverability and precision. We demonstrate that the Soft-PVTOL can be modeled using the Euler-Lagrange equations by assuming a constant curvature model in the aerial robot's arms. Such a mathematical model is presented in detail and can be extended to several constant curvature segments in each Soft-PVTOL arm. Moreover, we design a passivity-based control law that exploits the flexibility of the robot's arms. We solve the tracking control problem, proving that the error equilibrium globally exponentially converges to zero. The controller is tested in numerical simulations, demonstrating robust performance and ensuring the efficacy of the closed-loop system.
\end{abstract}

\end{frontmatter}

\section{Introduction}
The challenges inherent in modeling and controlling soft robots have prompted engineers to pursue simplified versions, particularly when the primary objective is achieving effective control \cite{armanini}. This pursuit is particularly relevant in soft mobile robots operating across six degrees of freedom, notably in aerial robotics. While significant strides have been made in modeling and controlling traditional aerial robots, including innovative configurations like convertible aircraft and fully actuated multi-rotors, the exploration of soft aerial robots remains limited. Despite their potential, research in this area is still in its infancy, with only a few studies addressing the prototypes and designs of soft aerial robots. In this endeavor, we aim to establish the fundamental principles of modeling and control for a soft aerial robot, leveraging the well-established paradigm of the Planar Vertical Take-Off and Landing (PVTOL) aircraft. The first concept of PVTOL aircraft was presented by Hauser et al. \cite{HAUSER1992665}. PVTOL prototypes have served as the cornerstone of aerial drone development, providing a robust platform for exploration and innovation. Our research aims to extend this foundation to soft aerial robotics. In this context, we present the Soft-PVTOL, where the agility of the soft arms lies in their ability to adapt dynamically, distributing forces effectively across multiple points of contact. Unlike rigid arms, which can only apply forces in predetermined directions, the soft-PVTOL arms adjust in real-time to changing conditions, enabling precise force control. This flexibility is demonstrated in our model through simulations where the arms' continuous curvature allows for smoother, more responsive maneuvers. The results highlight the system's capacity to perform complex, agile movements, particularly in scenarios where rigid systems would struggle due to their fixed degrees of freedom.
%
\begin{figure}[tb!]
\begin{center}
\includegraphics[width=\columnwidth]{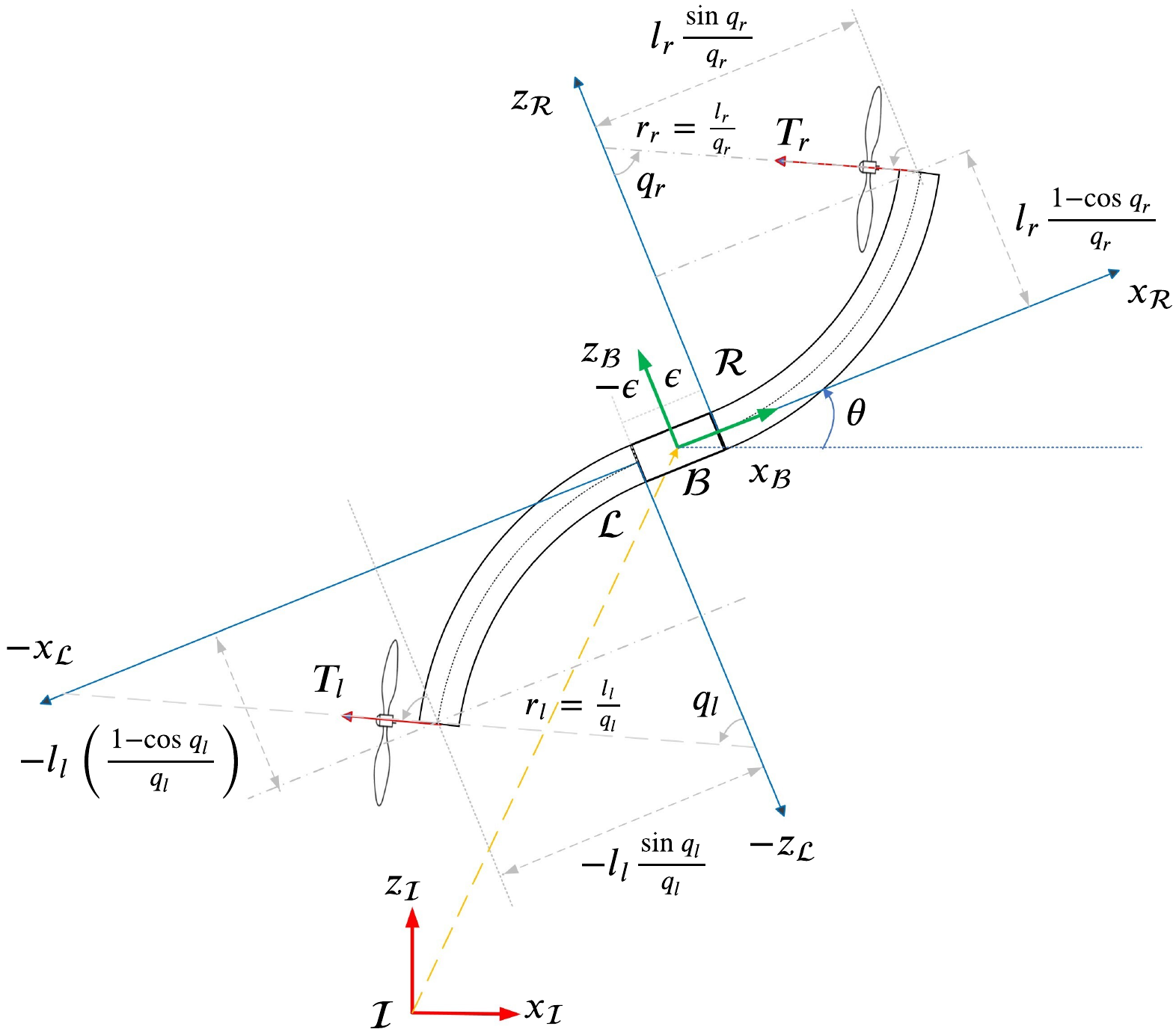}   
\caption{The concept of the Soft-PVTOL involves two soft arms modeled using the constant curvature approach. We adopt the convention that when the arm tilts below the $x$ axis in the frame, we assign a negative value, as illustrated in the case of $q_l$ in the diagram. This convention is already accounted for in \eqref{eq:coords}.} 
\label{fig:diagram}
\end{center}
\end{figure}
%
\subsection{State of the art}
In recent years, efforts to enhance drone maneuverability have led to exploring innovative configurations featuring tilting or folding rotor arms \cite{soro.2017.0120, 9851515, foldable, 8567932}. Incorporating foldable mechanisms offers numerous advantages, including size reduction, obstacle avoidance, navigation through narrow gaps, inspection of vertical surfaces, and object grasping and transportation. Consequently, adopting drones with adaptive morphology presents clear benefits in various operational scenarios. In this context, the potential of soft aerial robots offers unique advantages in maneuverability, flexibility, adaptability to complex environments, recovery from collisions in various directions, and even perching capabilities.

Recent research in the field of aerial soft robots has primarily focused on conventional multirotor configurations equipped with soft arms \cite{9762161, 9812044, HASHEMI2023160, 9635927}. While efforts have been made to design platforms \cite{9115993, soro.2022, 9837128}, and develop soft propellers \cite{9115983, 10103580}, these studies have yet to address the challenge of effectively modeling and controlling such soft aerial configurations. Furthermore, existing soft aerial robots typically feature passive soft arms that lack actuation capabilities, leaving the potential of fully actuated soft aerial robots unexplored.

The realm of modeling and control for soft robots continues to be a dynamic and challenging area of research \cite{SpongAnualReview}. Despite significant advancements, even the most popular soft robot, the manipulator, presents ongoing challenges \cite{soro.2017.0007, 10136424, Wang}. Exploring control strategies for soft aerial configurations remains relatively nascent, with only a few studies addressing this issue. For instance, Deng et al. \cite{pmlr-v155-deng21a} employ neural networks to model the dynamics of soft multi-rotors, enabling the generation of controllers that collectively manage locomotion.
\subsection{Contribution}
Exploring simplified models of soft aircraft is interesting because modeling and controlling soft robots are generally complex tasks. This led to the creation of the Soft-PVTOL, a compact version of a soft aerial robot. Its goal is to simplify complicated dynamics into an easy-to-understand mathematical model, which helps grasp the fundamental principles of the system dynamics. As with standard PVTOL models, the 2D formulation presented here serves as a critical step toward understanding and validating the fundamental behavior of the Soft-PVTOL system. Importantly, the principles and control methodologies developed in 2D extend naturally to 3D, just as with conventional PVTOLs and quadrotors, enabling a seamless transition to fully three-dimensional soft aerial robots.

By leveraging the idea of constant curvature of the rotor arms \cite{nnnikhil}, we demonstrate that the Soft-PVTOL can be modeled using the Euler-Lagrange approach widely used in robotics \cite{kellybook}. Such a model does not present singularities for all the feasible curvature changes. This approach is advantageous since one can exploit the system structure to propose nonlinear controllers. In this context, we also demonstrate that the Soft-PVTOL satisfies the passivity property \cite{sspongg}, and thus, we propose a passivity-based control demonstrating that the error equilibrium is globally exponentially stable. In designing the controller, we leverage the inherent soft properties of the arms, utilizing them as a virtual controller to govern the Soft-PVTOL's position. This innovative approach results in fully actuated dynamics, effectively decoupling orientation from position —a marked advantage over conventional PVTOL systems. Additionally, we address the control allocation problem by determining reference values for all four actuators within the system. This methodology holds promise for extension to soft multi-rotor configurations, promising further advancements in aerial robotics.
\subsection{Content}
The remainder of this paper is as follows. In Section \ref{sec:model}, we derive a mathematical model of the Soft-PVTOL, employing the Lagrange approach alongside the constant curvature paradigm for the robot's arms. Detailed computations are provided in an appendix. Additionally, this section shows the passivity property of the Soft-PVTOL. In Section \ref{sec:control}, we construct a passivity-based controller based on the well-known works of Spong, Ortega, and others, \cite{ortegapass,choprapassi,Spongbook,Ortegabook} to address the tracking control problem, establishing global exponential stability of error equilibrium. Here, we also tackle the control allocation challenge by determining references for system actuators—namely, motor thrusts and desired curvature values for the Soft-PVTOL arms. Section \ref{sec:sims} presents simulation results, validating the effectiveness of our proposed model and control strategy while addressing numerical challenges inherent to soft robots in simulation settings. Lastly, Section \ref{sec:conc} offers conclusions and outlines future directions, including the potential extension of this work to soft multi-rotors.
%
%
%
\section{The mathematical model of the Soft-PVTOL}\label{sec:model}
To capture the dynamics of the soft-PVTOL, we employ a reduced-order model that balances simplicity and accuracy while maintaining full actuation. This ensures robust and predictable control, allowing the direct application of established strategies. The flexibility of the soft arms enables the system to exhibit both underactuated and overactuated behavior, offering adaptability to various scenarios. While underactuation presents interesting challenges, our focus is to demonstrate the feasibility of soft aerial robotics to achieve overactuated dynamics by deforming their arms. This dual behavior preserves the advantages of soft robotic dynamics while ensuring stability and avoiding unnecessary complexities from rigid-link approximations.

The dynamic model considers the translation and rotation along the 2-D frame executed by the soft drone and the bending of the left and right rotor arms where we have assumed a constant curvature approach; see Fig. \ref{fig:diagram}. We define four types of main frames:
\begin{enumerate}
    \item The inertial frame $\mathcal{I}$.
    \item The body frame attached to the robot's center of gravity $\mathcal{B}$.
    \item The frames $\mathcal{L}_i$ for $i=\{ 0,1\}$ defining the beginning of the left-hand curve path $(\mathcal{L}_0)$ and the end of the left-hand curve path $(\mathcal{L}_1)$. 
    \item The frames $\mathcal{R}_i$ for $i=\{ 0,1\}$ defining the beginning of the right-hand curve path $(\mathcal{R}_0)$ and the end of the right-hand curve path $(\mathcal{R}_1)$.
\end{enumerate}

We first define the following generalized coordinates,
\begin{equation}\label{eq:general}
    \mathbf{q} = \begin{pmatrix}
        x_v & z_v & \theta & q_l & q_r
    \end{pmatrix}^{\intercal},
\end{equation}
where we consider the constant curvature approach for the left and right arms of the Soft-PVOL, where $q_l$ and $q_r$ are the curvatures of the left and right-hand side soft arms; see Fig. \ref{fig:thrust_right}. Soft manipulators in \cite{8722799, Cosimo:2020, Jones:2010, 10122091, 9844234} have used such an approach. It is worth noting that in a real Soft-PVTOL, these curvatures could be easily measured using simple Inertial Measurement Units (IMUs) positioned at the tip of the arms. This makes the proposed approach feasible and practical for real-world aerial Soft-PVTOL systems applications.

The vector of control inputs is defined as,
\begin{equation}
    \textbf{u} = \begin{pmatrix}
        \tau_x & \tau_z & \tau_{\theta} & \tau_l & \tau_r
    \end{pmatrix}^{\intercal}.
\end{equation}
Notice that we assume that there exists a mapping from $\textbf{u}$ to the vector of the actuators \begin{equation}
    \textbf{a} = \begin{pmatrix}
        T_l & T_r & \tau_l & \tau_r 
    \end{pmatrix}^{\intercal}.
\end{equation}
given by $\textbf{u} = B \textbf{a}$ that we must model later.

The system parameters are grouped into:
\begin{enumerate}
    \item $m$: mass of the robot; $I$: the total inertia moment of the robot
    \item $l_l$: length of the bending left-hand arm; $m_l$: mass of the left-hand rotor; $I_l$: the inertia moment of the left-hand arm
    \item $l_r$: length of the bending right-hand arm; $m_r$: mass of the right-hand rotor; $I_r$: the inertia moment of the right-hand arm.
    \item The rigid body of the drone is composed of a frame of length $2\epsilon$.
\end{enumerate}

\begin{assumption}
The center of mass of each soft arm is on its tip or equivalently on the origin of the frames $\mathcal{L}_1$ and $\mathcal{R}_1$.
\end{assumption}

We model the Soft-PVTOL using the Euler-Lagrange equation of motion as follows \cite{loria,ortegaloriabook,ortegaspong},
\begin{equation}\label{eq:EL}
     \frac{d}{dt} \left(\frac{\partial \mathcal{L}}{\partial \dot{\textbf{q}}} \right) - \frac{\partial \mathcal{L}}{\partial {\textbf{q}}} = \tau.
\end{equation}

To define the Lagrangian of the system, we first define the position of the Soft-PVTOL and that of the left and right-hand motor following Fig. \ref{fig:diagram} as follows,
\begin{equation}\label{eq:coords}
\begin{aligned}
p_{v} &= \begin{pmatrix}
    x_v  \\
    z_v 
\end{pmatrix}  , \\
p_{l} &= \begin{pmatrix}
    x_v  \\
    z_v 
\end{pmatrix} + \begin{pmatrix}
     -\epsilon - l_l\frac{ \sin{q_l}}{q_l} \\
     l_l\left[ \frac{1-\cos{q_l}}{q_l} \right]
\end{pmatrix}  , \\
p_{r} &= \begin{pmatrix}
    x_v  \\
    z_v 
\end{pmatrix} + \begin{pmatrix}
    \epsilon + l_r\frac{ \sin{q_r}}{q_r} \\
    l_r\left[ \frac{1-\cos{q_r}}{q_r} \right]
\end{pmatrix}
\end{aligned}
\end{equation}
where all the positions are w.r.t. the inertial frame. The time derivatives of the previous positions are given by,
\begin{equation}
\begin{aligned}
\dot{p}_{v} &= \begin{pmatrix}
    \dot{x}_v  \\
    \dot{z}_v 
\end{pmatrix}  , \\
\dot{p}_{l} &= \begin{pmatrix}
    \dot{x}_v  \\
    \dot{z}_v 
\end{pmatrix} + l_l\begin{pmatrix}
     \frac{\sin{q_l} - q_l\cos{q_l}}{q_l^2} \\
     \frac{q_l\sin{q_l}+\cos{q_l}-1}{q_l^2}
\end{pmatrix}\dot{q}_l  , \\
\dot{p}_{r} &= \begin{pmatrix}
    \dot{x}_v  \\
    \dot{z}_v 
\end{pmatrix} + l_r\begin{pmatrix}
    \frac{ q_r\cos{q_r} - \sin{q_r}}{q_r^2} \\
    \frac{q_r\sin{q_r} + \cos{q_r} - 1}{q_r^2}
\end{pmatrix}\dot{q}_r.
\end{aligned}
\end{equation}
The details of the computations of the elements in \eqref{eq:EL} can be verified in Appendix \ref{ap1}.
%
%
\subsection{Control input part}
Since the drone arms move freely due to their soft properties, combining two thrusts will generate forces in the coordinates $x-z$ and torque in $\theta$. The input control vector is:
\begin{equation}
\tau = 
\begin{pmatrix}
\tau_x  & \tau_z & \tau_{\theta} & \tau_l & \tau_r
\end{pmatrix}^{\intercal} .
\end{equation}
In the following, we find particular expressions for each of the above vector entries.
\subsubsection{Control input for the $(x,z)$ dynamics}\label{sssec:control_xz}
We first compute the total forces due to thrusts in the $x-z$ plane as follows. Refer to Fig. \ref{fig:thrust_right} where the right-hand side motor is depicted. Since the Soft-PVTOL position and orientation are measured w.r.t. the inertial frame, it follows that:
\begin{equation}\label{eq:pos_control}
\begin{aligned}
\begin{pmatrix}
    \tau_{x} \\
    \tau_{z}
\end{pmatrix} &=
    \underbrace{\begin{pmatrix}
        \cos{\theta} & - \sin{\theta} \\
        \sin{\theta} & \cos{\theta}
\end{pmatrix}}_{\textrm{Rotational matrix}} 
\Bigg[  
\underbrace{\begin{pmatrix}
    \cos{q_l}   & \sin{q_l} \\ 
    - \sin{q_l} & \cos{q_l}
\end{pmatrix}}_{\textrm{Rotational matrix$^{\intercal}$}}
    \underbrace{\begin{pmatrix}
    0 \\ 1    
    \end{pmatrix}}_{e_z}T_l \\
    &\hspace{30mm}+ 
    \underbrace{\begin{pmatrix}
    \cos{q_r}   & -\sin{q_r} \\ 
    \sin{q_r}   & \cos{q_r}
\end{pmatrix}}_{\textrm{Rotational matrix}}
    \begin{pmatrix}
    0 \\ 1
    \end{pmatrix}T_r
\Bigg] \\
&=\begin{pmatrix}
        \cos{\theta} & - \sin{\theta} \\
        \sin{\theta} & \cos{\theta}
\end{pmatrix} 
\underbrace{\begin{pmatrix}
    -T_r\sin{q_r} + T_l\sin{q_l} \\
    T_r\cos{q_r} + T_l\cos{q_l}
\end{pmatrix}}_{( \upsilon_x ,  \upsilon_z)^{\intercal}}.
\end{aligned}
\end{equation}
where $(\upsilon_x , \upsilon_z)$ are virtual control inputs. Notice that in the case of $q_l = q_r = 0$ the above expression results in $\begin{pmatrix}
        \tau_{x} \\
        \tau_{z}
    \end{pmatrix} = 
        \begin{pmatrix}
        -\sin{\theta}\\
        \cos{\theta}
    \end{pmatrix} (T_l + T_r)$
similar to the position control inputs in a conventional PVTOL.
\subsubsection{Control input for $\theta$}
Now, we compute the torque generated by both thrusts as follows. Let us consider Fig. \ref{fig:thrust_right}.
\begin{figure}[htb!]
\begin{center}
\includegraphics[width=\columnwidth]{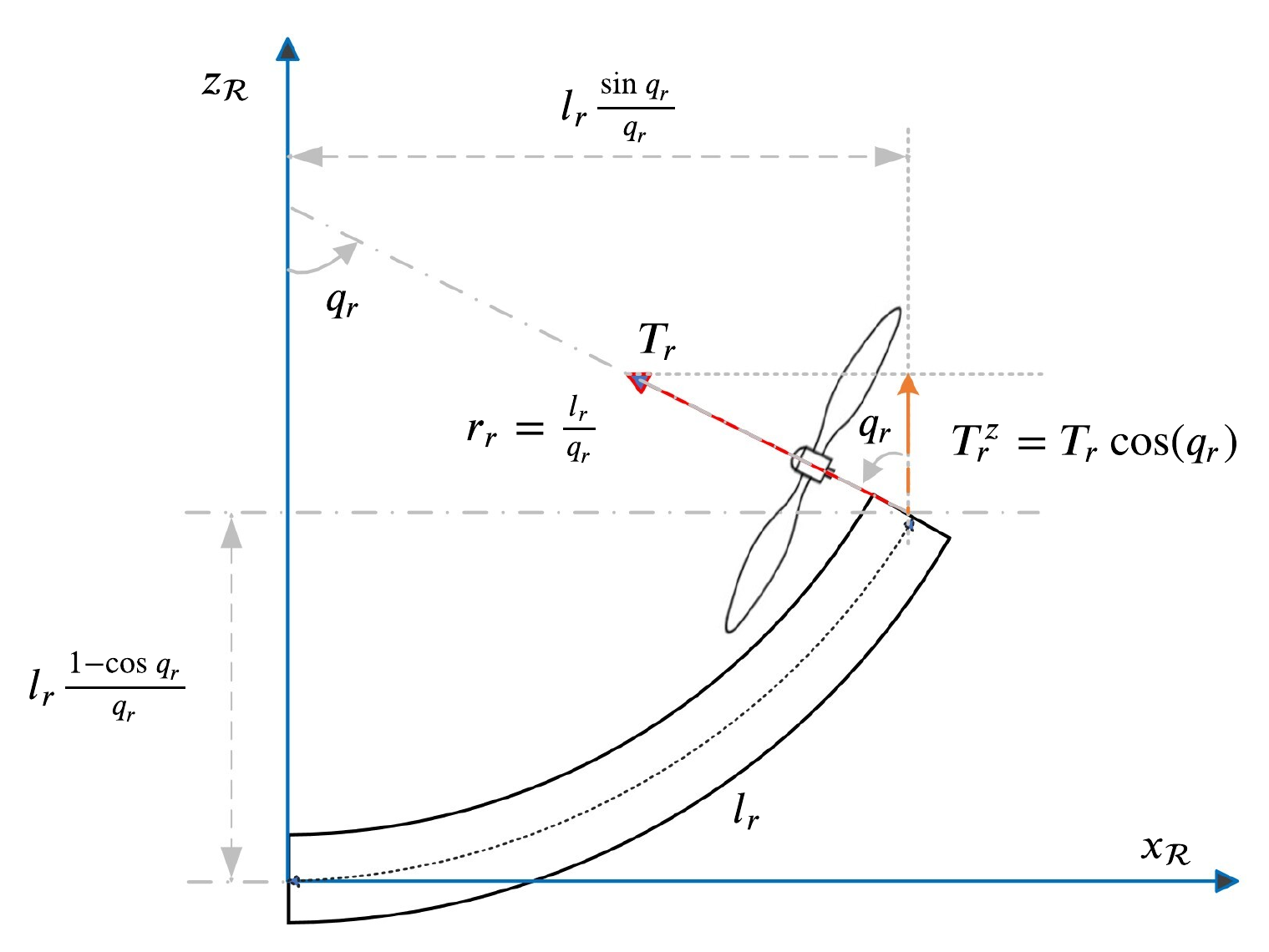}   
\caption{The thrust of the right-hand side motor. From the figure, it is easy to see that the torque generated by the right-hand-thrust $T_r$ is given by $\left(l_r\frac{\sin{qr}}{q_r}\right) \left( T_r\cos{q_r} \right)$. Please note that we have omitted the value of $\epsilon$ for simplicity.} 
    \label{fig:thrust_right}%
\end{center}
\end{figure}

According to Fig. \ref{fig:thrust_right}, the torque is given by,
\begin{equation}
    \tau_{\theta} =  \left[\textrm{proj}_{e_z}\begin{pmatrix}
        T_r^x \\ T_r^z
    \end{pmatrix} \right] \left[    l_r\frac{\sin{q_r}}{q_r} \right] -  \left[ \textrm{proj}_{e_z}\begin{pmatrix}
        T_l^x \\ T_l^z
    \end{pmatrix} \right] \left[ l_l\frac{\sin{q_l}}{q_l} \right] 
\end{equation}
where the scalar projection of vectors $(T_r^x, T_r^z)^{\intercal}$ and $(T_l^x, T_l^z)^{\intercal}$ onto $e_z = (0,1)^{\intercal}$ is given by
\begin{equation}
\begin{aligned}
\textrm{proj}_{e_z}\begin{pmatrix}
        T_r^x \\ T_r^z
    \end{pmatrix} &= \left( \frac{\begin{pmatrix}
        T_r^x \\ T_r^z
    \end{pmatrix} \cdot e_z}{\Vert e_z \Vert^2} \right), \\
\textrm{proj}_{e_z}\begin{pmatrix}
        T_l^x \\ T_l^z
    \end{pmatrix} &= \left( \frac{\begin{pmatrix}
        T_l^x \\ T_l^z
    \end{pmatrix} \cdot e_z}{\Vert e_z \Vert^2} \right).
\end{aligned}
\end{equation}
Since $T_r^z = T_r\cos{q_r}$ and $T_l^z = T_l\cos{q_l}$ are respectively the components of the RHS and LHS thrusts, it follows that\footnote{For simplicity, in what follows we will assume that $\epsilon = 0$.}:
\begin{equation}\label{eq:att_control_input}
    \tau_{\theta} =  l_r   T_r \cos{q_r}   \frac{\sin{q_r}}{q_r} -  l_l   T_l\cos{q_l}    \frac{\sin{q_l}}{q_l} 
\end{equation}
and notice that if $q_l=q_r=0$ the torque $\tau_{\theta}$ is similar to that of the conventional PVTOL.
%
\subsubsection{Control inputs for the curvature angles $(q_l,q_r)$}
Finally, the $(\tau_l,\tau_r)$ correspond to the tendons control input and are assumed to be independent of the current state variables.
%
\subsubsection{Control input in vector form}
\begin{equation}
\tau = 
\begin{pmatrix}
\tau_x  \\
\tau_z \\
\tau_{\theta} \\
\tau_l \\
\tau_r
\end{pmatrix} = 
\begin{pmatrix}
\begin{pmatrix}
 \cos{\theta} & - \sin{\theta} \\
 \sin{\theta} & \cos{\theta}
\end{pmatrix}
\underbrace{\begin{pmatrix}
 -T_r\sin{q_r} + T_l\sin{q_l}  \\
 T_r\cos{q_r} + T_l\cos{q_l}
\end{pmatrix}}_{(\upsilon_x, \upsilon_z)^{\intercal}} \\
l_r   T_r \cos{q_r}   \frac{\sin{q_r}}{q_r} -  l_l   T_l\cos{q_l}    \frac{\sin{q_l}}{q_l}   \\
\tau_l \\
\tau_r
\end{pmatrix}.
\end{equation}
%
\subsection{Matrix form}
Let us rewrite the previous equations in the matrix form, \cite{Spongbook}:
\begin{equation}\label{eq:system_Euler-Lagrange}
D(\textbf{q})\ddot{\textbf{q}} + C(\textbf{q}, \dot{\textbf{q}})\dot{\textbf{q}} + g(\textbf{q}) = \tau,
\end{equation}
where 
\begin{equation}
    \begin{aligned}
   \textbf{q} &= \begin{pmatrix}
        x_v & z_v & \theta & q_l & q_r
    \end{pmatrix}^{\intercal}, \\
    \dot{\textbf{q}} &= \begin{pmatrix}
        \dot{x}_v & \dot{z}_v & \dot{\theta} & \dot{q}_l & \dot{q}_r
    \end{pmatrix}^{\intercal}.      
    \end{aligned}
\end{equation}
After several computations, one gets:
\begin{equation}\label{eq:inertia}
D(\textbf{q}) = 
\begin{pmatrix}
d_{11} & 0 & 0 & d_{14} & d_{15} \\
0 & d_{22} & 0 & d_{24} & d_{25} \\
0 & 0 & d_{33} & 0 & 0 \\
d_{41} & d_{42} & 0 & d_{44} & 0 \\
d_{51} & d_{52} & 0 & 0 & d_{55}
\end{pmatrix},
\end{equation}
where 
\begin{equation}
\begin{aligned}
d_{11} &= d_{22} = m+m_l + m_r \\
d_{33} &= I \\
d_{44} &= l_l^2m_l\left(\frac{1}{q_l^2}\right) + 2l_l^2m_l\left(\frac{1}{q_l^4}\right) -2l_l^2m_l\left(  \frac{\cos{q_l}}{q_l^4} \right) \\
&\hspace{10mm}-2l_l^2m_l\left(  \frac{\sin{q_l}}{q_l^3} \right) + I_l  \\
d_{55} &= l_r^2m_r\left(\frac{1}{q_r^2}\right) + 2l_r^2m_r\left(\frac{1}{q_r^4}\right) -2l_r^2m_r\left(  \frac{\cos{q_r}}{q_r^4} \right) \\
&\hspace{10mm}-2l_r^2m_r\left(  \frac{\sin{q_r}}{q_r^3} \right) + I_r
\end{aligned}
\end{equation}
and 
\begin{equation}
\begin{aligned}
d_{14} &= d_{41} = -  l_lm_l \left(   \frac{\cos{q_l}}{q_l}  -   \frac{\sin{q_l}}{q_l^2}  \right) \\
d_{15} &= d_{51} = l_rm_r \left(   \frac{\cos{q_r}}{q_r}  -   \frac{\sin{q_r}}{q_r^2}  \right)\\
d_{24} &= d_{42} = -l_lm_l \left(  \frac{1}{q_l^2} -  \frac{\cos{q_l}}{q_l^2}  - \frac{\sin{q_l}}{q_l}  \right) \\
d_{25} &= d_{52} = -l_rm_r \left( \frac{1}{q_r^2} - \frac{\cos{q_r}}{q_r^2}  -   \frac{\sin{q_r}}{q_r}  \right)
\end{aligned}
\end{equation}
It can be verified that there is no singularity in any of the terms of $D(\textbf{q})$, even when $q_i \rightarrow 0$ where $i = \{l,r \}$. One can confirm this assertion by calculating the limit of each term $d_{mn}$ in $D(\textbf{q})$ as $q_i \rightarrow 0$, noting that the outcome consistently remains a constant and, in numerous instances, approaches zero. This is true for all the $d_{mn}$.

The Coriolis matrix, after some simplifications, is given by:
\begin{equation}
C(\textbf{q}, \dot{\textbf{q}}) = 
\begin{pmatrix}
0 & 0 & 0 & c_{14} & c_{15} \\
0 & 0 & 0 & c_{24} & c_{25} \\
0 & 0 & 0 & 0 & 0 \\
0 & 0 & 0 & c_{44} & 0  \\
0 & 0 & 0 & 0 & c_{55}
\end{pmatrix} 
\end{equation}
where
\begin{equation}
\begin{aligned}
c_{14} &= l_lm_l \left( \frac{q_l\sin{q_l} + \cos{q_l}}{q_l^2} \right) \dot{q_l} \\
&\hspace{10mm}+ l_lm_l \left( \frac{q_l\cos{q_l} - 2\sin{q_l}}{q_l^3} \right) \dot{q_l} \\
c_{15} &= -  l_rm_r \left( \frac{q_r\sin{q_r} + \cos{q_r}}{q_r^2} \right) \dot{q_r} \\
&\hspace{10mm}- l_rm_r \left( \frac{q_r\cos{q_r} - 2\sin{q_r}}{q_r^3} \right) \dot{q_r} ,
\end{aligned}
\end{equation}
\begin{equation}
\begin{aligned}
c_{24} &=  2l_lm_l \left(\frac{1}{q_l^3}\right)\dot{q_l} - l_lm_l \left( \frac{q_l\sin{q_l} + 2\cos{q_l}}{q_l^3} \right) \dot{q_l} \\
&\hspace{10mm}+ l_lm_l \left( \frac{q_l\cos{q_l} - \sin{q_l}}{q_l^2} \right) \dot{q}_l \\
c_{25} &=  2l_rm_r \left(\frac{1}{q_r^3}\right)\dot{q}_r - l_rm_r \left( \frac{q_r\sin{q_r} + 2\cos{q_r}}{q_r^3} \right) \dot{q}_r \\
&\hspace{10mm} + l_rm_r \left( \frac{q_r\cos{q_r} - \sin{q_r}}{q_r^2} \right) \dot{q}_r ,
\end{aligned}
\end{equation}
\begin{equation}
\begin{aligned}
c_{44} &= - l_l^2m_l \left(\frac{1}{q_l^3}\right)\dot{q}_l - 4l_l^2m_l \left(\frac{1}{q_l^5}\right)\dot{q}_l  \\
&\hspace{10mm} + l_l^2m_l \left( \frac{q_l\sin{q_l} + 4\cos{q_l}}{q_l^5} \right) \dot{q}_l \\
&\hspace{20mm}- l_l^2m_l \left( \frac{q_l\cos{q_l} - 3\sin{q_l}}{q_l^4} \right) \dot{q}_l \\
c_{55} &= - l_r^2m_r \left(\frac{1}{q_r^3}\right)\dot{q}_r - 4l_r^2m_r \left(\frac{1}{q_r^5}\right)\dot{q}_r \\
&\hspace{10mm}+ l_r^2m_r \left( \frac{q_r\sin{q_r} + 4\cos{q_r}}{q_r^5} \right) \dot{q}_r \\
&\hspace{20mm} - l_r^2m_r \left( \frac{q_r\cos{q_r} - 3\sin{q_r}}{q_r^4} \right) \dot{q}_r,
\end{aligned}
\end{equation}
where it can be shown that there are no singularities in any term of $C(\textbf{q}, \dot{\textbf{q}})$.

And finally,
\begin{equation}\label{eq:gGravity}
    g(\mathbf{q}) = \begin{pmatrix}
        0  \\
        g(m+m_l+m_r)  \\
        0 \\
        (g m_l l_l)\left(  \frac{q_l\sin{q_l} + \cos{q_l} - 1}{q_l^2} \right)  \\
        (gm_r l_r)\left(  \frac{q_r\sin{q_r} + \cos{q_r} - 1}{q_r^2}\right)  
    \end{pmatrix}.
\end{equation}
There are also no singularities for all the terms of $g(\textbf{q})$.
\subsection{Properties of the dynamic equations}
We introduce the following parameter equations:
\begin{equation}
\begin{alignedat}{3}
\Theta_1 &= m + m_l + m_r , \qquad & \Theta_2 &= I , \qquad & \Theta_3 &= l_lm_l, \\
\Theta_4 &= l_rm_r, \qquad & \Theta_5 &= I_l, \qquad & \Theta_6 &= I_r
\end{alignedat}
\end{equation}
With the above notation, the inertia matrix is rewritten as, 
{\small
\begin{equation}\label{eq:inertia_new}
    D(\mathbf{q}) = \begin{pmatrix}
        \Theta_1 & 0 & 0 & \Theta_3 \mathscr{D}_{1} & \Theta_4 \mathscr{D}_{3}   \\
        0 & \Theta_1 & 0 & \Theta_3 \mathscr{D}_{2}    & \Theta_4 \mathscr{D}_{4}  \\
        0 & 0 & \Theta_2 & 0 & 0 \\
        \Theta_3 \mathscr{D}_{1}   & \Theta_3 \mathscr{D}_{2} & 0 & l_l\Theta_3\mathscr{D}_{5} + \Theta_5  & 0 \\
         \Theta_4 \mathscr{D}_{3}  & \Theta_4 \mathscr{D}_{4}    & 0 & 0 & l_r\Theta_4 \mathscr{D}_{6}  + \Theta_6
    \end{pmatrix},
\end{equation}}
where
\begin{equation}
    \begin{aligned}
        \mathscr{D}_{1} &= \left(  \frac{\sin{q_l} - q_l\cos{q_l}}{q_l^2}  \right)  \\ 
        \mathscr{D}_{2} &= \left(   \frac{ \cos{q_l} + q_l\sin{q_l} -1  }{q_l^2}  \right) \\
        \mathscr{D}_{3} &= \left(  \frac{q_r\cos{q_r} - \sin{q_r} }{q_r^2} \right)  \\
        \mathscr{D}_{4} &= \left(  \frac{ \cos{q_r} + q_r\sin{q_r} - 1 }{q_r^2}  \right) \\
        \mathscr{D}_{5} &= \left( \frac{q_l^2 + 2 - 2\cos{q_l} - 2q_l\sin{q_l}  }{q_l^4} \right)\\
        \mathscr{D}_{6} &= \left( \frac{ q_r^2 + 2 - 2\cos{q_r} - 2q_r\sin{q_r} }{q_r^4} \right).
    \end{aligned}
\end{equation}
\subsubsection{Positiveness of the inertia matrix}
We first need to prove that $D(\mathbf{q})$ is positive definite for all the possible values of $\mathbf{q}$.
\begin{lemma}\label{lemm:matrix}
The matrix $D(\mathbf{q})$ in \eqref{eq:inertia_new} is positive definite.
\end{lemma}
\begin{proof}
See Appendix \ref{app:lemma_matrix}.
\end{proof}
%
%
\subsubsection{Skew symmetry property}
On the other hand, using the Christoffel symbols as suggested by \cite{Spongbook}, the Coriolis matrix is expressed as follows:
\begin{equation}\label{eq:Coriolis_reduced}
    C(\textbf{q}, \dot{\textbf{q}}) = \begin{pmatrix}
        0 & 0 & 0 & \Theta_3 \mathscr{C}_{1} \dot{q_l} & - \Theta_4 \mathscr{C}_{3} \dot{q_r} \\
        0 & 0 & 0 &  \Theta_3 \mathscr{C}_{2} \dot{q}_l &  \Theta_4 \mathscr{C}_{4} \dot{q}_r \\
        0 & 0 & 0 & 0 & 0 \\
        0 & 0 & 0 & - 2l_l\Theta_3 \mathscr{C}_{5} \dot{q}_l & 0  \\
        0 & 0 & 0 & 0 & - 2l_r\Theta_4 \mathscr{C}_{6} \dot{q}_r t
    \end{pmatrix},
\end{equation}
where
\begin{equation}
\begin{aligned}
    \mathscr{C}_{1} &= \frac{ [q_l^2-2]\sin{q_l} + 2q_l\cos{q_l}  }{q_l^3} \\
    \mathscr{C}_{2} &=  \frac{[q_l^2-2]\cos{q_l} - 2q_l\sin{q_l} + 2}{q_l^3} \\
    \mathscr{C}_{3} &=  \frac{[q_r^2-2]\sin{q_r} + 2q_r\cos{q_r}}{q_r^3}\\
    \mathscr{C}_{4} &= \frac{[q_r^2-2]\cos{q_r} - 2q_r\sin{q_r} + 2 }{q_r^3} \\
    \mathscr{C}_{5} &=  \frac{\left(q_l\cos{\left(\frac{q_l}{2}\right)} - 2\sin{\left(\frac{q_l}{2} \right)}\right)^2}{q_l^5} \\
    \mathscr{C}_{6} &=  \frac{\left(  q_r\cos{\left(\frac{q_r}{2}\right)} - 2\sin{\left(\frac{q_r}{2}\right)} \right)^2}{q_r^5} 
\end{aligned}
\end{equation}
and $g(\mathbf{q})$ is defined as in \eqref{eq:gGravity}.
\begin{lemma}\label{eq:SkewMatrix}
The matrix $\dot{D}(\textbf{q}) - 2C(\textbf{q}, \dot{\textbf{q}})$ is skew-symmetric.
\end{lemma}
\begin{proof}
See Appendix \ref{app:lemmaSkew}.
\end{proof}
%
%
\subsubsection{Passivity property}\label{sssec:passivity}
Following a similar analysis as in \cite{Spongbook}, the total energy of the Soft-PVTOL given by the sum of kinetic and potential energy is:
\begin{equation}
    H = \frac{1}{2}\dot{\textbf{q}}^{\intercal}D(\textbf{q}) \dot{\textbf{q}} + P(\textbf{q}).
\end{equation}
The first-time derivative of $H$ along the trajectories of the Soft-PVTOL considering the skew-symmetry property \eqref{eq:skew} is,
\begin{equation}
    \dot{H} = \dot{\textbf{q}}^{\intercal} \tau.
\end{equation}
Integrating the last equation, the passivity property follows.
%
%
%
\section{Control}\label{sec:control}
The primary advantage of the Soft-PVTOL over its conventional counterpart lies in its exceptional maneuverability. Its ability to decouple attitude and position dynamics is particularly noteworthy, overcoming the inherent underactuation in conventional PVTOL systems. Additionally, deflecting its arms allows it to apply forces and torques to the environment with greater agility than traditional PVTOL designs. Thus, one desires to decouple the attitude $\theta$ from its position $(x_v,z_v)$ dynamics. Now, notice that the control input vector $\tau \in \mathbb{R}^5$ has four actuators given by:
\begin{equation}
\begin{aligned}
    &\textrm{Motors thrusts: } (T_l, T_r) \\
    &\textrm{Tendons for the arms: } (\tau_l, \tau_r),
\end{aligned}
\end{equation}
which correspond to the thrusts and the tendons for the left and right motors. Thus, we have four actuators to control three DOFs, and the system is overactuated.

For the control design, we should consider the following:
\begin{enumerate}
    \item We aim to decouple the orientation from the position, allowing for independent control of each, thus achieving heightened maneuverability in the Soft-PVTOL.
    \item In our control design process, we initially introduce a control law based on passivity properties widely investigated in \cite{Spongbook} for the control input $\tau$ in \eqref{eq:system_Euler-Lagrange}. Subsequently, employing a control allocation approach, we determine the reference values for the four actuators mentioned earlier.
\end{enumerate}
To achieve point 1) above, the Soft-PVTOL position will be indirectly controlled by $(q_l, q_r)$. So, we assume that $(q_l, q_r)$ are virtual controllers for the position dynamics.

We illustrate the concept of the proposed controller with a diagram, as shown in Fig. \ref{fig:controldiagram}.
\begin{figure*}[htb!]
\begin{center}
\includegraphics[width=0.9\textwidth]{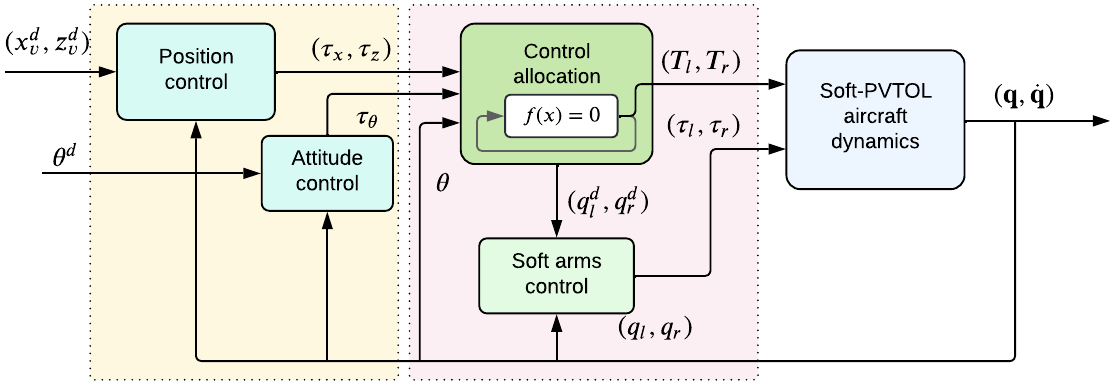}   
\caption{The block diagram illustrates the closed-loop system, where it is assumed that all system states $(\mathbf{q}, \dot{\mathbf{q}})$ are available for feedback—a standard assumption in aerial robot control.} 
\label{fig:controldiagram}
\end{center}
\end{figure*}

\begin{theorem}[Passivity-based control \cite{Spongbook}]\label{th:main}
Consider the Soft-PVTOL system modeled by equation \eqref{eq:system_Euler-Lagrange}. Also, consider Assumption \ref{ass:approx} and Lemmas \ref{lemm:matrix} and \ref{eq:SkewMatrix} together with the passivity property presented in Section \ref{sssec:passivity}. Let us consider the control law,
\begin{equation}\label{eq:passControl}
   \tau =  D(\mathbf{q})a + C(\mathbf{q}, \dot{\mathbf{q}})\upsilon + g(\mathbf{q}) - Kr,
\end{equation}
where 
\begin{equation}\label{eq:controlterms}
    \begin{aligned}
        \upsilon &= \dot{\mathbf{q}}^d - \Lambda \tilde{\mathbf{q}} \\
        a &= \ddot{\mathbf{q}}^d - \Lambda \dot{\tilde{\mathbf{q}}} \\
        r &= \dot{\mathbf{q}} - \upsilon
    \end{aligned}
\end{equation}
and $K,\Lambda$ are positive definite diagonal matrices. Then, the equilibrium $e=0$ of the tracking error 
\begin{equation}\label{eq:error__}
    e = \begin{pmatrix}
        \tilde{\mathbf{q}} \\ \dot{\tilde{\mathbf{q}}} 
    \end{pmatrix} = 
    \begin{pmatrix}
        \mathbf{q} - \mathbf{q}^d \\
        \dot{\mathbf{q}} - \dot{\mathbf{q}}^d
    \end{pmatrix}
\end{equation}
is globally exponentially stable.
\end{theorem}
%
\begin{proof}
First, we substitute the control law \eqref{eq:passControl} into the Soft-PVTOL system \eqref{eq:system_Euler-Lagrange} to obtain the closed-loop representation:
\begin{equation}
D(\textbf{q})\dot{r} + C(\textbf{q}, \dot{\textbf{q}})r + Kr = 0.
\end{equation}
Then, we propose the Lyapunov function candidate, 
\begin{equation}\label{eq:lyap}
    V(r,\tilde{q}) = \frac{1}{2}r^{\intercal}D(\textbf{q})r + \tilde{q}^{\intercal} \Lambda K \tilde{q}
\end{equation}
and we evaluate the closed-loop system trajectories along $V$ to get,
\begin{equation}
    \dot{V} = -r^{\intercal} Kr +2\tilde{q}^{\intercal}\Lambda K \dot{\tilde{q}} + \frac{1}{2} r^{\intercal}\left( \dot{D}(\textbf{q}) - 2 C(\mathbf{q}, \dot{\mathbf{q}}) \right) r.
\end{equation}
The last equation can be reduced by claiming Lemma \ref{eq:SkewMatrix} and the definition of $r$ in \eqref{eq:controlterms}. Thus, it follows that,
\begin{equation}\label{eq:lyapn_finall}
    \dot{V} = - \tilde{q}^{\intercal} \Lambda^{\intercal} K \Lambda \tilde{q} - \dot{\tilde{q}}^{\intercal} K \dot{\tilde{q}} = -e^{\intercal}\begin{pmatrix}
        \Lambda^{\intercal} K \Lambda & 0 \\
        0 & K
        \end{pmatrix}e.
\end{equation}
To prove exponential stability, consider the following analysis. Considering Assumption \ref{ass:approx}, the inertia matrix $D(\textbf{q})$ in \eqref{eq:inertia_new} has constant norm bounds $\Vert D(\textbf{q}) \Vert \leq \alpha \in \mathbb{R}^+$ for all feasible values of $q_l$ and $q_r$ given by \eqref{eq:bounds_ql_qr}. Notice that $V(r,\tilde{q})$ in \eqref{eq:lyap} can be expressed as,
\begin{equation}
    V(r,\tilde{q}) = - \begin{pmatrix}
            r^{\intercal} & \tilde{q}^{\intercal}
        \end{pmatrix}\underbrace{\begin{pmatrix}
        \frac{1}{2}D(\textbf{q}) & 0 \\
        0 & \Lambda K 
        \end{pmatrix}}_{P}\begin{pmatrix}
            r \\ \tilde{q}       
        \end{pmatrix}.
\end{equation}
Thus, let us now represent $\dot{V}(\dot{\tilde{q}}, \tilde{q})$ in \eqref{eq:lyapn_finall} in the coordinates $(r,\tilde{q})$ using \eqref{eq:controlterms} and \eqref{eq:error__} by substituting $\dot{\tilde{q}} = r - \Lambda \tilde{q}$ in \eqref{eq:lyapn_finall}:
\begin{equation}
\begin{aligned}
    \dot{V}(r,\tilde{q}) &= - \tilde{q}^{\intercal} \Lambda^{\intercal} K \Lambda \tilde{q} - \left( r^{\intercal} - \tilde{q}^{\intercal} \Lambda^{\intercal} \right) K \left( r - \Lambda \tilde{q} \right) \\
    &= -r^{\intercal} K r - \tilde{q}^{\intercal}\Lambda^{\intercal} K \Lambda \tilde{q} + 2r^{\intercal} K \Lambda \tilde{q} \\
    &= - \begin{pmatrix}
            r^{\intercal} & \tilde{q}^{\intercal}
        \end{pmatrix}\underbrace{\begin{pmatrix}
        K & K \Lambda \\
        K \Lambda & 2\Lambda^{\intercal} K \Lambda
        \end{pmatrix}}_{Q}\begin{pmatrix}
            r \\ \tilde{q}       
        \end{pmatrix}.
\end{aligned}
\end{equation}
We verify that the block matrix $Q$ in the above equation is positive definite by noting that all the principal leading minors are positive. Since $K$ is positive definite by definition, we must verify that $\det{(Q)} >0$:
\begin{equation}
\begin{aligned}
    \det{(Q)} &= \det{(K)}\det{\left( 2\Lambda^{\intercal} K \Lambda - K \Lambda K^{-1} K \Lambda \right)} \\
    &= \det{(K)}\det{(\Lambda)} \det{(\Lambda K)},
\end{aligned}
\end{equation}
which is clearly positive since $K$ and $\Lambda$ are positive definite diagonal matrices. Now, notice that,
\begin{equation}
    \dot{V} = -z^{\intercal}Qz \leq -\lambda_{\textrm{min}}\{ Q \} z^{\intercal} z \leq - \frac{\lambda_{\textrm{min}}\{ Q \}}{\lambda_{\textrm{max}}\{ P \}} z^{\intercal}Pz = - \rho V
\end{equation}
where $z = \begin{pmatrix} r & \tilde{q} \end{pmatrix}^{\intercal}$, and $\rho = \frac{\lambda_{\textrm{min}}\{ Q \}}{\lambda_{\textrm{max}}\{ P \}} >0$. Therefore, given that $V$ is radially unbounded and $\dot{V} \leq - \rho V$ holds, it implies the global exponential stability of the equilibrium point $e=0$.
\end{proof}
\subsection{Control allocation}\label{ssec:alloc}
Putting together \eqref{eq:pos_control} and \eqref{eq:att_control_input}, results in,
\begin{equation}\label{eq:taus_rot}
\begin{aligned}
    \begin{pmatrix}
        \tau_x \\ \tau_z \\ \tau_{\theta}
    \end{pmatrix} &=
    \begin{pmatrix}
        \cos{\theta} & - \sin{\theta} & 0 \\
        \sin{\theta} & \cos{\theta} & 0 \\
        0            & 0            & 1
    \end{pmatrix} \times \\
    &\hspace{5mm} \underbrace{\begin{pmatrix}
    -T_r\sin{q_r} + T_l\sin{q_l} \\
    T_r\cos{q_r} + T_l\cos{q_l} \\   l_r   T_r \cos{q_r}   \frac{\sin{q_r}}{q_r} -  l_l   T_l\cos{q_l}    \frac{\sin{q_l}}{q_l}
    \end{pmatrix}}_{(\upsilon_x , \upsilon_z , \tau_{\theta})^{\intercal}}.
\end{aligned}
\end{equation}
To recover the real control inputs $(T_r, T_l)$ and the desired angles $(q_r,q_l)$, we proceed as follows. First, notice that the vector $\mathbf{u} = (\upsilon_x , \upsilon_z , \tau_{\theta})^{\intercal}$ can be rewritten as:
\begin{equation}\label{eq:ctrls_alloc}
    \underbrace{\begin{pmatrix}
        \upsilon_x \\ \upsilon_z \\ \tau_{\theta}
    \end{pmatrix}}_{\mathbf{u}} = 
    \underbrace{\begin{pmatrix}
    -1 & 0 & 1 & 0  \\
    0 & 1 & 0 & 1  \\
    0 & l_r\frac{\sin{q_r}}{q_r} &  0 & -l_l\frac{\sin{q_l}}{q_l}
    \end{pmatrix}}_{A}
    \underbrace{\begin{pmatrix}
        T_r\sin{q_r} \\ T_r\cos{q_r} \\ T_l\sin{q_l} \\ T_l\cos{q_l}
    \end{pmatrix}}_{\mathbf{v}}.
\end{equation}
Using the Moore–Penrose inverse and after some computations, we get:
\begin{equation}
    \mathbf{v} = A^{+}\mathbf{u}
\end{equation}
where
\begin{equation}
    A^{+} =  
    \begin{pmatrix}
        -\frac{1}{2} & 0 & 0 \\
        0 &  \frac{l_lq_r\sin{q_l}}{l_lq_r\sin{q_l} + l_rq_l\sin{q_r}}  & \frac{q_lq_r}{l_lq_r\sin{q_l} + l_rq_l\sin{q_r}} \\
        \frac{1}{2} & 0 & 0 \\
        0 & \frac{l_rq_l\sin{q_r}}{l_lq_r\sin{q_l} + l_rq_l\sin{q_r}} & - \frac{q_lq_r}{l_lq_r\sin{q_l} + l_rq_l\sin{q_r}}
    \end{pmatrix},
\end{equation}
and we noticed that such an equation is not well defined when $q_l=q_r=0$, which will throw numerical problems when one tries to solve it for $(T_l, T_r, q_l, q_r)$. To avoid that situation, let us make the following assumption.
\begin{assumption}\label{ass:approx}
The Soft-PVOL will achieve bounded angles $(q_l,q_r)$ such that,
\begin{equation}\label{eq:bounds_ql_qr}
    \vert q_l \vert \leq \pi, \qquad \vert q_r \vert \leq \pi.
\end{equation}
Moreover, in that region, we approximate:
\begin{equation}
    \frac{\sin{q_l}}{q_l} 	\approx \cos{\left(\frac{q_l}{2} \right)}, \qquad \frac{\sin{q_r}}{q_r} 	\approx \cos{\left(\frac{q_r}{2}\right)}.
\end{equation}
\end{assumption}
The approximation claimed in Assumption \ref{ass:approx} makes sense since, in practical situations, the angles $(q_l,q_r)$ hardly reach angles near to $\pi$. The approximation can be seen graphically in the Fig. \ref{fig:approx}.
\begin{figure}[tb!]
\begin{center}
\includegraphics[width=\columnwidth]{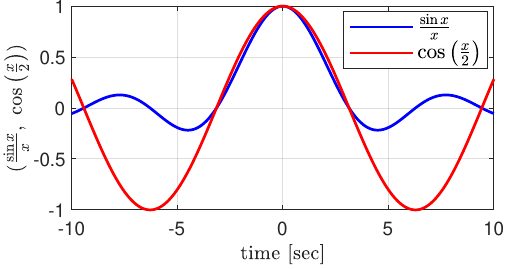}
\caption{Function approximation described in Assumption \ref{ass:approx}.} 
    \label{fig:approx}%
\end{center}
\end{figure}
Thus, considering Assumption \ref{ass:approx} we describe \eqref{eq:ctrls_alloc} as follows, 
\begin{equation}
    \mathbf{u} = \bar{A}\mathbf{v}
\end{equation}
where
\begin{equation}
    \bar{A} = 
    \begin{pmatrix}
    -1 & 0 & 1 & 0  \\
    0 & 1 & 0 & 1  \\
    0 & l_r\cos{\left(\frac{q_r}{2} \right)} &  0 & -l_l\cos{\left(\frac{q_l}{2}\right)}
    \end{pmatrix},
\end{equation}
and using the Moore-Penrose inverse, it results in,
\begin{equation}\label{eq:TrTlangles}
\begin{aligned}
    \mathbf{v} &= \bar{A}^{+}\mathbf{u} \\
    &= \begin{pmatrix}
-\frac{1}{2}\upsilon_x \\
\frac{l_l \cos\left(\frac{q_l}{2}\right)}{l_l \cos\left(\frac{q_l}{2}\right) + l_r \cos\left(\frac{q_r}{2}\right)}\upsilon_z + \frac{1}{l_l \cos\left(\frac{q_l}{2}\right) + l_r \cos\left(\frac{q_r}{2}\right)}\tau_{\theta} \\
\frac{1}{2}\upsilon_x \\
\frac{l_r \cos\left(\frac{q_r}{2}\right)}{l_l \cos\left(\frac{q_l}{2}\right) + l_r \cos\left(\frac{q_r}{2}\right)}\upsilon_z - \frac{1}{l_l \cos\left(\frac{q_l}{2}\right) + l_r \cos\left(\frac{q_r}{2}\right)}\tau_{\theta}
\end{pmatrix}
\end{aligned}
\end{equation}
where
\begin{equation}
\bar{A}^{+} = \begin{pmatrix}
-\frac{1}{2} & 0 & 0 \\
0 & \frac{l_l \cos\left(\frac{q_l}{2}\right)}{l_l \cos\left(\frac{q_l}{2}\right) + l_r \cos\left(\frac{q_r}{2}\right)} & \frac{1}{l_l \cos\left(\frac{q_l}{2}\right) + l_r \cos\left(\frac{q_r}{2}\right)} \\
\frac{1}{2} & 0 & 0 \\
0 & \frac{l_r \cos\left(\frac{q_r}{2}\right)}{l_l \cos\left(\frac{q_l}{2}\right) + l_r \cos\left(\frac{q_r}{2}\right)} & -\frac{1}{l_l \cos\left(\frac{q_l}{2}\right) + l_r \cos\left(\frac{q_r}{2}\right)}
\end{pmatrix}
\end{equation}
where the RHS of \eqref{eq:TrTlangles} is well defined even in $q_l=q_r=0$. And from \eqref{eq:taus_rot} notice that,
\begin{equation}
    \begin{pmatrix}
        \upsilon_x \\ \upsilon_z \\ \tau_{\theta}
    \end{pmatrix} = 
        \begin{pmatrix}
        \cos{\theta} & - \sin{\theta} & 0 \\
        \sin{\theta} & \cos{\theta} & 0 \\
        0            & 0            & 1
    \end{pmatrix}^{\intercal}
    \begin{pmatrix}
        \tau_x \\ \tau_z \\ \tau_{\theta}
    \end{pmatrix} 
\end{equation}
and then
\begin{equation}\label{eq:vxvz}
    \begin{aligned}
        \upsilon_x &= \tau_x\cos{\theta} + \tau_z\sin{\theta}  \\
        \upsilon_z &= -\tau_x\sin{\theta} + \tau_z \cos{\theta}.
    \end{aligned}
\end{equation}
Now, using \eqref{eq:TrTlangles} we solve for $(T_l,T_r,q_l,q_r)$ using simple trigonometric functions $\cos^2{x} + \sin^2{x}=1$, and $\tan{x}=\frac{\sin{x}}{\cos{x}}$, as follows,
\begin{equation}\label{eq:control_inputs_des}
    \begin{aligned}
    T_l &= \sqrt{ \left( \frac{1}{2}\upsilon_x \right)^2 +  \left( \frac{l_r\cos{ \left(\frac{q_r}{2} \right)} \upsilon_z - \tau_{\theta}}{l_l \cos\left(\frac{q_l}{2}\right) + l_r \cos\left(\frac{q_r}{2}\right)} \right)^2   } \\
    T_r &= \sqrt{ \left( \frac{1}{2}\upsilon_x \right)^2 +  \left( \frac{l_l\cos{ \left(\frac{q_l}{2} \right)} \upsilon_z + \tau_{\theta}}{l_l \cos\left(\frac{q_l}{2}\right) + l_r \cos\left(\frac{q_r}{2}\right)} \right)^2   } \\
    q_l &= \arctan{ \left(  \frac{\left( \upsilon_x \right) \left( l_l \cos\left(\frac{q_l}{2}\right) + l_r \cos\left(\frac{q_r}{2}\right) \right)}{ \left( 2 \right) \left( l_r\cos{ \left(\frac{q_r}{2} \right)} \upsilon_z - \tau_{\theta} \right) } \right)} \\
    q_r &= \arctan{ \left( - \frac{\left( \upsilon_x \right) \left( l_l \cos\left(\frac{q_l}{2}\right) + l_r \cos\left(\frac{q_r}{2}\right) \right)}{ \left( 2 \right) \left( l_l\cos{ \left(\frac{q_l}{2} \right)} \upsilon_z + \tau_{\theta} \right) } \right)}.
    \end{aligned} 
\end{equation}
It is noteworthy that the equations above, coupled with \eqref{eq:vxvz}, necessitate numerical methods for solution due to their nonlinear nature and dependence on $(q_l,q_r)$. \textit{These $(q_l,q_r)$ angles will henceforth be referred to as the desired angles for the controller, denoted by $(q_l^d,q_r^d)$}. Thus, we solve \eqref{eq:control_inputs_des} numerically as,
\begin{equation}\label{eq:solve_num}
\begin{aligned}
    f(x) &=  
    \begin{pmatrix}
    T_l - \sqrt{ \left( \frac{1}{2}\upsilon_x \right)^2 +  \left( \frac{l_r\cos{ \left(\frac{q_r^d}{2} \right)} \upsilon_z - \tau_{\theta}}{l_l \cos\left(\frac{q_l^d}{2}\right) + l_r \cos\left(\frac{q_r^d}{2}\right)} \right)^2   } \\
    T_r - \sqrt{ \left( \frac{1}{2}\upsilon_x \right)^2 +  \left( \frac{l_l\cos{ \left(\frac{q_l^d}{2} \right)} \upsilon_z + \tau_{\theta}}{l_l \cos\left(\frac{q_l^d}{2}\right) + l_r \cos\left(\frac{q_r^d}{2}\right)} \right)^2   } \\
    q_l^d - \arctan{ \left(  \frac{\left( \upsilon_x \right) \left( l_l \cos\left(\frac{q_l^d}{2}\right) + l_r \cos\left(\frac{q_r^d}{2}\right) \right)}{ \left( 2 \right) \left( l_r\cos{ \left(\frac{q_r^d}{2} \right)} \upsilon_z - \tau_{\theta} \right) } \right)} \\
    q_r^d - \arctan{ \left( - \frac{\left( \upsilon_x \right) \left( l_l \cos\left(\frac{q_l^d}{2}\right) + l_r \cos\left(\frac{q_r^d}{2}\right) \right)}{ \left( 2 \right) \left( l_l\cos{ \left(\frac{q_l^d}{2} \right)} \upsilon_z + \tau_{\theta} \right) } \right)}
    \end{pmatrix} \\
    &= 0
    \end{aligned}
\end{equation}
where $x = (T_l, T_r, q_l^d, q_r^d)$.
%
\section{Simulations}\label{sec:sims}
The simulation experiments were performed using the Matlab Simulink ode5 (Dormand-Prince) solver with a fixed-step size of $0.01$ seconds. To solve \eqref{eq:solve_num}, we use the Algebraic block of Matlab Simulink with initial guess $x= [20, 10, 0.4, -0.2]$.

The system parameters are: $m = 5$, $m_l = m_r = \frac{1}{5}m$, $I=1$, $I_l = I_r = 0.1$, $l_l = l_r = 0.5$, and $g = 9.8$.

The control parameters are: $k_{px} = 3$, $k_{dx} = 0.85k_{px}$, $k_{pz} = 3$, $k_{dz} = 0.85k_{pz}$, $k_{p\theta} = 15$, $k_{d\theta} = 0.7k_{p\theta}$, $k_{pq_{l}} = 30$, $k_{dq_{l}} = 0.7k_{pq_{l}}$, and $k_{pq_{r}} = 30$, $k_{dq_{r}} = 0.7k_{pq_{r}}$. $K_P = \diag{(k_{px} , \ k_{pz} , \ k_{p\theta} , \ k_{pq_{l}} , \ k_{pq_{r}} )}$, $K_D = \diag{(k_{dx} , \ k_{dz} , \ k_{d\theta} , \ k_{dq_{l}} , \ k_{dq_{r}} )}$. Also, $\lambda_1 = 1$, $\lambda_2 = 1$, $\lambda_3 = 5$, $\lambda_4 = 10$, $\lambda_5 = 10$, and $\Lambda = \diag{(\lambda_1 , \ \lambda_2 , \ \lambda_3 , \ \lambda_4 , \ \lambda_5)}$. 

The initial conditions are 
\begin{equation}
\begin{aligned}
        \mathbf{q}(0) &= 
    \begin{pmatrix}
        x_v(0) & z_v(0) & \theta(0) & q_l(0) & q_r(0)
    \end{pmatrix}^{\intercal} \\
    &= \begin{pmatrix}
        5 & 0 & 0.2\pi & 0.01\pi & -0.15\pi 
    \end{pmatrix}^{\intercal} \\
    \dot{\mathbf{q}}(0) &= 
    \begin{pmatrix}
        \dot{x}_v(0) & \dot{z}_v(0) & \dot{\theta}(0) & \dot{q}_l(0) & \dot{q}_r(0)
    \end{pmatrix}^{\intercal} \\
    &= \begin{pmatrix}
        0 & 0 & 0 & 0 & 0 
    \end{pmatrix}^{\intercal}.
\end{aligned}
\end{equation}
The desired trajectories for the $(x_v-z_v)$ coordinates and the pitch angle $\theta$ are given by:
\begin{equation}\label{eq:traj_des}
\begin{aligned}
    &x_v^d =  4\sin{\left(\frac{1}{2}t\right)}, \quad
    z_v^d = 3\cos{\left(-\frac{1}{10} t\right)} + 4, \quad
    \theta^d = 0 \\
    &\dot{x}_v^d =  2\cos{\left(\frac{1}{2}t\right)}, \quad
    \dot{z}_v^d = -0.3\sin{\left(-\frac{1}{10} t\right)} , \quad
    \dot{\theta}^d = 0.
 \end{aligned}   
\end{equation}
Notice that the Soft-PVTOL must track a time-varying $x_v-z_v$ trajectory depicted in Fig. \ref{fig:traj} while maintaining a roll angle equal to zero. To achieve that, the soft extremities must track a controlled trajectory. 
\begin{figure}[tb!]
\centering
\includegraphics[width=0.9\columnwidth]{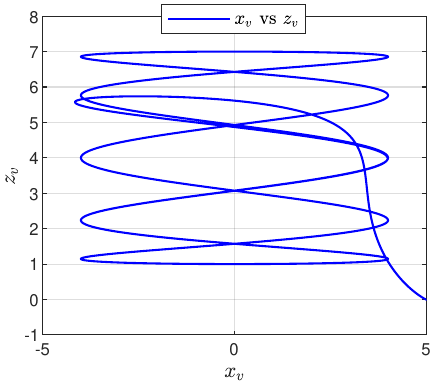}
\caption{The $x-z$ plot illustrates the 2D trajectory tracked by the Soft-PVTOL. It is important to note that while in a conventional PVTOL, the $x$ displacements are associated with tilting angles of $\theta$, in this particular configuration, the roll angle remains fixed at zero throughout.}
\label{fig:traj}
\end{figure} 
In this manner, it becomes possible to control both the position and the orientation of the Soft-PVTOL separately. This is different from conventional PVTOL and multi-rotors, where changing the position requires adjusting the orientation because of the underactuated nature of the PVTOL system. 

In Fig. \ref{fig:errors}, we observe the position error vector $\tilde{\mathbf{q}}(t)$ and the velocity error vector $\dot{\tilde{\mathbf{q}}}(t)$. Notably, all error states converge to zero exponentially, and the convergence rate can be adjusted by tuning the control gains. It is observed that the convergence of the error states $(\tilde{q}_{x_v}, \tilde{q}_{z_v})$ occurs slower than the rest of the states. This is expected due to a time-scale separation among the position, orientation, and arm dynamics subsystems, \cite{6580063}.
\begin{figure*}[htb!]
\centering
\includegraphics[width=0.75\textwidth]{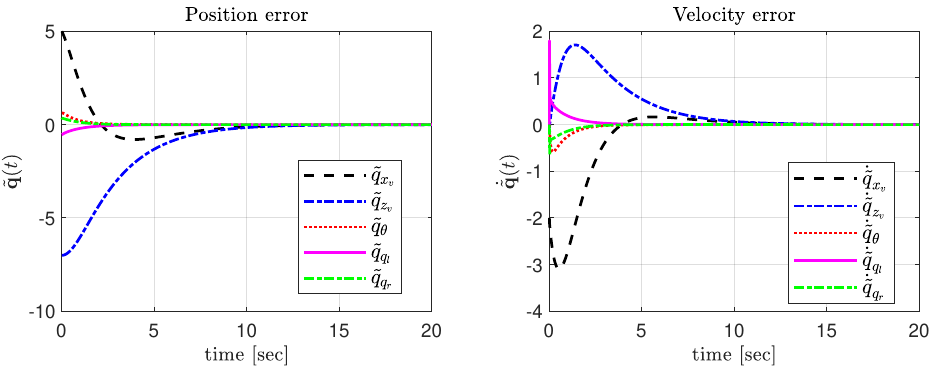}
\caption{Position and velocity error states in the closed-loop system.}
\label{fig:errors}
\end{figure*} 

For a clearer view of the Soft-PVTOL's convergence to the desired position outlined in \eqref{eq:traj_des}, we separately plot the actual position $(x_v, z_v)$ against its desired values $(x_v^d, z_v^d)$ and the roll angle $\theta$ against its desired value $\theta^d$. These are presented on the left-hand side of Fig. \ref{fig:pos_att}. Similarly, we depict the corresponding velocities on the right-hand side of the same figure. This figure makes it more evident how the position changes over time while the roll angle remains fixed at zero.
\begin{figure*}[htb!]
\centering
\includegraphics[width=0.75\textwidth]{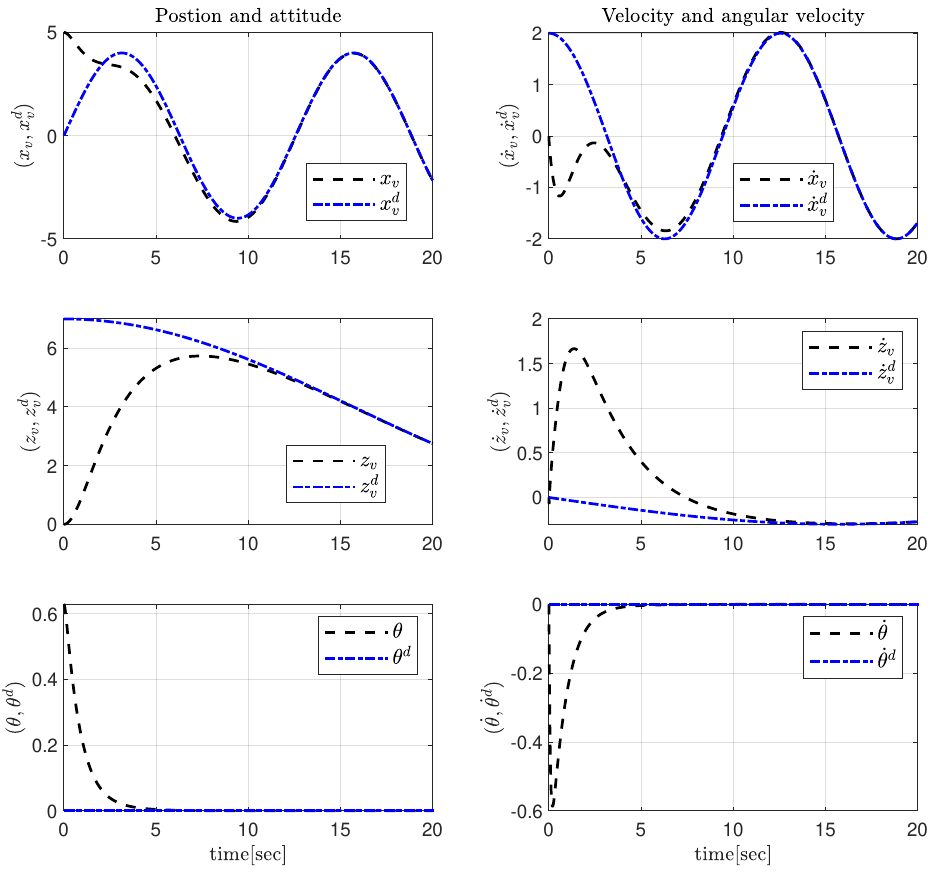}
\caption{The left-hand side displays the pose of the Soft-PVTOL in the closed-loop response, accompanied by reference values. Meanwhile, the right-hand side illustrates the velocities and angular velocities alongside their respective references.}
\label{fig:pos_att}
\end{figure*} 

Analyzing the Euler-Lagrange equations \eqref{eq:EL}, we observe that the $\tau \in \mathbb{R}^5$ vector encompasses all control inputs for the Soft-PVTOL aircraft. The components of this control input vector $\tau$ are illustrated in Fig. \ref{fig:control_inputs}. Notably, some elements of this control vector, such as $(\tau_x, \tau_z)$, are virtual, as discussed in Sections \ref{sssec:control_xz} and \ref{ssec:alloc}.

Upon examining Fig. \ref{fig:control_inputs}, it becomes apparent that $\tau_z$ significantly surpasses $\tau_x$, as anticipated. This discrepancy is expected, given that the primary force exertion is directed upwards to counteract the weight of the Soft-PVTOL aircraft.
\begin{figure*}[htb!]
\centering
\includegraphics[width=0.75\textwidth]{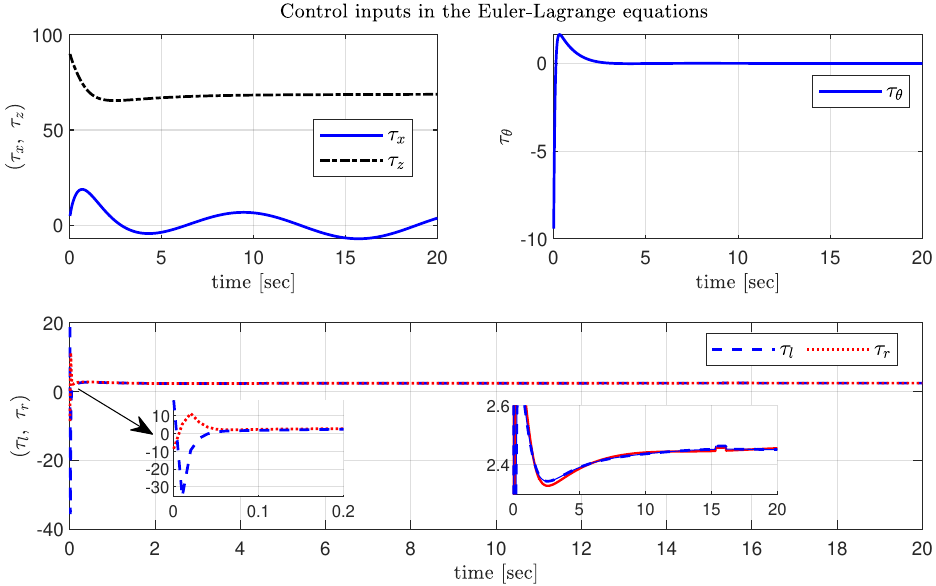}
\caption{The control inputs in the Euler-Lagrange equation \eqref{eq:EL}. In the Euler-Lagrange equation, the control inputs for the left and right-hand side arms, denoted as $(\tau_l, \tau_r)$, exhibit oscillations around $2.5$ to attain the desired curvature angles. An interesting observation is the minor variation around $t=16$ seconds. This arises from computing limits when $q_l$ or $q_r$ approach zero to mitigate singularity complexities in the simulation; see Section \ref{ssec:singular} for details of it.}
\label{fig:control_inputs}
\end{figure*} 

Once the control values $\tau$ are determined using Theorem \ref{th:main}, we proceed to derive the actual control inputs: the motor thrusts $(T_l, T_r)$ and the desired curvature angles for the arms $(q_l^d, q_r^d)$. These are illustrated in Fig. \ref{fig:thrusts_qlqr}. Additionally, within the same figure, we observe the convergence of the arm curvature angles $q_l$ and $q_r$ to their desired values $q_l^d$ and $q_r^d$, respectively. It is important to recall that $(T_l, T_r, q_l^d, q_r^d)$ are calculated as outlined in \eqref{eq:control_inputs_des}. Specifically, we solve \eqref{eq:solve_num} in Matlab Simulink using the \textit{algebraic constraint} block with the initial guess described at the beginning of this section.
\begin{figure*}[htb!]
\centering
\includegraphics[width=0.75\textwidth]{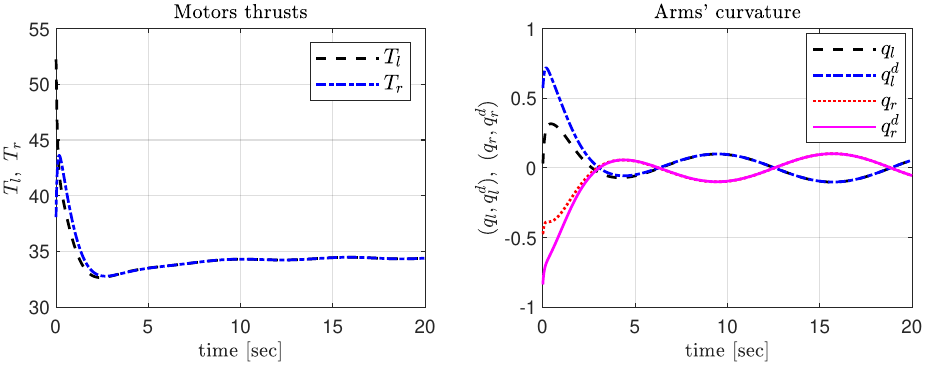}
\caption{The thrusts generated by the left and right motors are illustrated on the left-hand side, while the curvature and their corresponding desired values for the left and right soft arms are depicted on the right-hand side.}
\label{fig:thrusts_qlqr}
\end{figure*} 

\subsection{Avoiding numerical issues}\label{ssec:singular}
While the system itself does not contain any singularities, the numerical simulation can encounter errors due to the computational approach used by the simulator.

For instance, consider the entries of the matrix  $D(\mathbf{q})$. Throughout the inertia matrix $D(\mathbf{q})$, and indeed in the rest of the matrices of the system, various terms are multiplied by $\frac{1}{q_l^n}$ or $\frac{1}{q_r^n}$, where $n = \{1, \cdots , 5 \}$. In Matlab Simulink, such terms can lead to errors in the simulation because the equations are computed numerically, term by term, rather than analytically.

To address this issue and avoid numerical errors, we analytically compute the limit of each term as $q_l$ and $q_r$ approach zero. For example, in \eqref{eq:inertia_new}, we compute the limit of the entries of the matrix $D(\mathbf{q})$ as $q_l$ and $q_r$ approach zero, specifically $\lim_{q_l\to 0 }\mathscr{D}_{1} = 0$. We then substitute this calculated value directly into the simulation when $\vert q_l \vert < \delta$, where $\delta$ is a small positive number. In this simulation, we set $\delta = 0.1$. Outside of this region, $\mathscr{D}_{1}$ retains its usual value as described in $D(\mathbf{q})$.

This approach ensures that the simulation accurately handles terms involving small values of $q_l$ and $q_r$, preventing numerical errors that may arise otherwise.
%
%
\section{Conclusions and future directions}\label{sec:conc}
In this paper, we have presented the Soft-PVTOL aircraft for the first time, which is the version of the PVTOL with soft arms. We have demonstrated that the mathematical model of such an aerial robot can be represented by the Euler-Lagrange equations assuming constant curvature in the soft arms of the robot. Obtaining a concise mathematical model with the Euler-Lagrange approach opens the door to propose different well-known and mature control strategies to stabilize the error dynamics of the Soft-PVTOL. Moreover, the model can be extended to a quadrotor or even more complex multi-rotor aerial soft robots since a multi-rotor is an extension of a PVTOL; such property remains in the Soft-PVTOL. It is interesting that despite the complexity of the model, it does not present any singularities for all the possible curvature values in the left and right soft arms.
%

In contrast to conventional PVTOL and multi-rotor systems, where adjustments in position invariably influence orientation due to inherent underactuation, Soft-PVTOL distinguishes itself by releasing this characteristic. This distinctive feature heralds a new era of maneuverability and precision and signifies a substantial leap forward in aerial robotics technology. We have shown that including soft arms renders the complete system fully actuated, unlike the conventional PVTOL. Leveraging this characteristic, we utilize the curvature of the arms as a virtual controller to attain arbitrary feasible desired positions and orientations. Consequently, the Soft-PVTOL successfully decouples orientation from position dynamics, allowing for the independent tracking of trajectories for $(x_v,z_v)$ and $\theta$. Unlike traditional methods that approximate soft arms using a series of rigid links, our model captures the fluidity of motion inherent to soft robotics, preserving both the structural flexibility and control precision necessary for aerial systems.

We have validated the passivity property for the Soft-PVTOL and leveraged it by implementing a passivity-based controller. Nevertheless, intriguing control challenges persist for this category of soft aerial robots, such as robust and adaptive control strategies. While our control strategy leverages state-of-the-art methods, a key contribution lies in the novel application of the Euler-Lagrange framework to model the Soft-PVTOL system. This approach ensures the integration of well-established control techniques while maintaining the intrinsic softness and continuous deformation of the robot’s arms.

The simulation results exhibit the convergence of all system states, validating the effectiveness of the proposed approach. However, a challenge arose with Matlab Simulink's computation methodology during the numerical simulation. Simulink calculates expressions term by term, leading to singularities when the curvature of one arm reaches zero. Nevertheless, analytically, we have demonstrated that such singularities do not exist. To address this discrepancy, we analytically computed the limit of each entry in the Euler-Lagrange equations matrices as the curvature of the arms approaches zero. This enabled us to successfully simulate the closed-loop system, demonstrating the exponential convergence of all error states.

While the configuration may resemble a folding arm design, the key distinction lies in the continuous deformability of the soft arms, offering advantages that rigid structures cannot achieve. The softness is not just structural but functional, enabling smooth force distribution and dynamic adjustments during flight. This allows the Soft-PVTOL to transition seamlessly between underactuated and overactuated behaviors, enhancing both maneuverability and resilience in unpredictable environments. Unlike revolute joints, which impose discrete constraints, the soft arms provide finer control by adapting forces across multiple points. This adaptability supports precise, energy-efficient maneuvers and safer interactions, making the soft structure essential for applications where flexibility and performance are critical.

A future direction lies in path planning, where the focus will be on deducing desired trajectories, considering potential combinations and interactions between complex maneuvers and the desired curvature angles for the arms. This involves understanding how different maneuvers and arm configurations can be combined and coordinated to achieve optimal paths and movements.

There is a lot of open research in modeling, control, design, and construction in this field. This work opens the door for future investigation in more complex scenarios where the constant curvature assumption is relaxed or when the arms can twist. An immediate continuation of this research is the extension for soft quadrotors and their design and construction.

\onecolumn
\appendix{\textbf{Appendix}}
\section{Computations of the Euler-Lagrange equation}\label{ap1}
First, we must to compute the Lagrangian, $\mathcal{L} = \mathcal{K} - \mathcal{P}$, where the Kinetic energy has translational $\mathcal{K}_T$ and rotational $\mathcal{K}_R$ components:
\begin{equation}
    \mathcal{K} = \underbrace{\frac{1}{2}m \dot{p}_v^{\intercal}\dot{p}_v + \frac{1}{2}m_l \dot{p}_{l}^{\intercal}\dot{p}_{l} + \frac{1}{2}m_r \dot{p}_{r}^{\intercal}\dot{p}_{r}}_{\mathcal{K}_T} +  \underbrace{\frac{1}{2}I\dot{\theta}^2 + \frac{1}{2}I_l\dot{q_l}^2 +  \frac{1}{2}I_r\dot{q_r}^2
 }_{\mathcal{K}_R},
\end{equation}
where the translational kinetic energy is rewritten as follows,
\begin{equation}
\begin{aligned} 
  \mathcal{K}_T &= \frac{1}{2}m\left( \dot{x}_v^2 + \dot{z}_v^2 \right) + \frac{1}{2}m_{l}\dot{x}_v^{2}+\frac{1}{2}m_{l}\dot{z}_v^{2}+\frac{1}{2}%
l_{l}^{2}\frac{m_{l}}{q_{l}^{2}}\dot{q}_l^{2} +l_{l}^{2}\frac{m_{l}}{q_{l}^{4}} \dot{q}_l^{2} -l_{l}^{2}m_{l}\frac{\cos
q_{l}}{q_{l}^{4}}\dot{q}_l^{2}-l_{l}^{2}m_{l}\frac{\sin q_{l}}{q_{l}^{3}}\dot{q}_l^{2} - l_{l}\frac{m_{l}}{q_{l}^{2}}\dot{q}_l\dot{z}_v \\
&\hspace{5mm}  -l_{l}m_{l}\frac{\cos
q_{l}}{q_{l}}\dot{q}_l\dot{x}_v + l_{l}m_{l}\frac{\cos q_{l}}{q_{l}^{2}}\dot{q}_l\allowbreak \dot{z}_v+l_{l}%
m_{l}\frac{\sin q_{l}}{q_{l}^{2}}\dot{q}_l\dot{x}_v  + l_{l}m_{l}\frac{\sin q_{l}}{q_{l}}\dot{q}_l\dot{z}_v + \frac{1}{2}m_{r}\dot{x}_v^{2}+\frac{1}{2}m_{r}\dot{z}_v^{2}+\frac{1}{2}l_{r}^{2}\frac{m_{r}}{q_{r}^{2}}\dot{q}_r^{2}\\
&\hspace{10mm} + l_{r}^{2}\frac{m_{r}}{q_{r}^{4}}\dot{q}_r^{2}-l_{r}^{2}m_{r}\frac{\cos
q_{r}}{q_{r}^{4}}\dot{q}_r^{2}-l_{r}^{2}m_{r}\frac{\sin q_{r}}{q_{r}^{3}}\dot{q}_r^{2}  - l_{r}\frac{m_{r}}{q_{r}^{2}}\dot{q}_r\dot{z}_v + l_{r}m_{r}\frac{\cos
q_{r}}{q_{r}}\dot{q}_r\dot{x}_v+l_{r}m_{r}\frac{\cos q_{r}}{q_{r}^{2}}\dot{q}_r\dot{z}_v  \\
&\hspace{15mm} - l_{r}m_{r}\frac{\sin q_{r}}{q_{r}^{2}}\dot{q}_r\dot{x}_v + l_{r}m_{r}\frac{\sin q_{r}}{q_{r}}\dot{q}_r\dot{z}_v .
\end{aligned}
\end{equation}
The potential energy is
\begin{equation}
    \mathcal{P} =  mgz_v + m_l g \left(z_v +  l_l \frac{1-\cos{q_l}}{q_l} \right) + m_r g \left( z_v +  l_r  \frac{1-\cos{q_r}}{q_r} \right).
\end{equation}
Thus, the Lagrangian is
\begin{equation}
\begin{aligned}
\mathcal{L} &= \mathcal{K} - \mathcal{P} \\
&= \frac{1}{2}m \dot{p}_v^{\intercal}\dot{p}_v + \frac{1}{2}m_l \dot{p}_{l}^{\intercal}\dot{p}_{l}  +  \frac{1}{2}m_r \dot{p}_{r}^{\intercal}\dot{p}_{r} + \frac{1}{2}I\dot{\theta}^2 + \frac{1}{2}I_l\dot{q_l}^2  +  \frac{1}{2}I_r\dot{q_r}^2 - (gz_v)\left( m + m_l + m_r \right)\\
&\hspace{5mm}  - (g)\left( m_l l_l \frac{1-\cos{q_l}}{q_l} + m_r l_r \frac{1-\cos{q_r}}{q_r}\right) 
\end{aligned}
\end{equation}
that can be expressed as follows,
\begin{equation}
\begin{aligned}
\mathcal{L} &= \frac{1}{2}m\left( \dot{x}_v^2 + \dot{z}_v^2 \right) + \frac{1}{2}m_{l}\dot{x}_v^{2} + \frac{1}{2}m_{l}\dot{z}_v^{2} + \frac{1}{2}
l_{l}^{2}\frac{m_{l}}{q_{l}^{2}}\dot{q}_l^{2}  + l_{l}^{2}\frac{m_{l}}{q_{l}^{4}}
\dot{q}_l^{2}  -l_{l}^{2}m_{l}\frac{\cos
q_{l}}{q_{l}^{4}}\dot{q}_l^{2}-l_{l}^{2}m_{l}\frac{\sin q_{l}}{q_{l}^{3}}\dot{q}_l^{2} - l_{l}
\frac{m_{l}}{q_{l}^{2}}\dot{q}_l\dot{z}_v 
\\
&\hspace{5mm} - l_{l}m_{l}\frac{\cos
q_{l}}{q_{l}}\dot{q}_l\dot{x}_v  + l_{l}m_{l}\frac{\cos q_{l}}{q_{l}^{2}}\dot{q}_l\dot{z}_v + l_{l}
m_{l}\frac{\sin q_{l}}{q_{l}^{2}}\dot{q}_l\dot{x}_v  + l_{l}m_{l}\frac{\sin q_{l}}{q_{l}}
\dot{q}_l\dot{z}_v + \frac{1}{2}m_{r}\dot{x}_v^{2} + \frac{1}{2}m_{r}\dot{z}_v^{2} + \frac{1}{2}l_{r}^{2}\frac{m_{r}}{q_{r}^{2}}\dot{q}_r^{2} \\
&\hspace{10mm}  + l_{r}^{2}\frac{m_{r}}{q_{r}^{4}}\dot{q}_r^{2} - l_{r}^{2}m_{r}\frac{\cos
q_{r}}{q_{r}^{4}}\dot{q}_r^{2} - l_{r}^{2}m_{r}\frac{\sin q_{r}}{q_{r}^{3}}\dot{q}_r^{2} - l_{r}
\frac{m_{r}}{q_{r}^{2}}\dot{q}_r\dot{z}_v   + l_{r}m_{r}\frac{\cos
q_{r}}{q_{r}}\dot{q}_r\dot{x}_v + l_{r}m_{r}\frac{\cos q_{r}}{q_{r}^{2}}\dot{q}_r\dot{z}_v    \\
&\hspace{15mm} - l_{r}
m_{r}\frac{\sin q_{r}}{q_{r}^{2}}\dot{q}_r\dot{x}_v + l_{r}m_{r}\frac{\sin q_{r}}{q_{r}}\dot{q}_r\dot{z}_v + \frac{1}{2}I\dot{\theta}^2 + \frac{1}{2}I_l\dot{q_l}^2  +  \frac{1}{2}I_r\dot{q_r}^2 \\
&\hspace{20mm} - (gz_v)\left( m + m_l + m_r \right) - (g)\left( m_l l_l \frac{1-\cos{q_l}}{q_l} + m_r l_r \frac{1-\cos{q_r}}{q_r}\right).
\end{aligned}
\end{equation}

From \eqref{eq:general} it follows that
\begin{equation}
    \dot{\textbf{q}} = \begin{pmatrix}
        \dot{x}_v & \dot{z}_v & \dot{\theta} & \dot{q}_l & \dot{q}_r
    \end{pmatrix}^{\intercal}
\end{equation}
and then we compute,
\begin{equation}
 \frac{\partial \mathcal{L}}{\partial \dot{\textbf{q}}} = \begin{pmatrix}
        \frac{\partial \mathcal{L}}{\partial \dot{x}_v} &
        \frac{\partial \mathcal{L}}{\partial \dot{z}_v} &
        \frac{\partial \mathcal{L}}{\partial \dot{\theta}} &
        \frac{\partial \mathcal{L}}{\partial \dot{q}_l} &
        \frac{\partial \mathcal{L}}{\partial \dot{q}_r} 
    \end{pmatrix}^{\intercal}
\end{equation}
where
\begin{equation}
\begin{aligned}
\frac{\partial \mathcal{L}}{\partial \dot{x}_v} &=  m\dot{x}_v + m_l\dot{x}_v - l_lm_l\frac{\cos{q_l}}{q_l}\dot{q}_l + l_lm_l\frac{\sin{q_l}}{q_l^2}\dot{q}_l    + m_r\dot{x}_v + l_rm_r\frac{\cos{q_r}}{q_r}\dot{q}_r - l_rm_r\frac{\sin{q_r}}{q_r^2}\dot{q}_r \\
\frac{\partial \mathcal{L}}{\partial \dot{z}_v} &=  m\dot{z}_v + m_l\dot{z}_v - l_l\frac{m_l}{q_l^2}\dot{q}_l + l_lm_l\frac{\cos{q_l}}{q_l^2}\dot{q}_l + l_lm_l\frac{\sin{q_l}}{q_l}\dot{q}_l   +  m_r\dot{z}_v - l_r \frac{m_r}{q_r^2}\dot{q}_r + l_rm_r\frac{\cos{q_r}}{q_r^2}\dot{q}_r + l_rm_r\frac{\sin{q_r}}{q_r}\dot{q}_r \\
\frac{\partial \mathcal{L}}{\partial \dot{\theta}} &= I\dot{\theta} \\
\frac{\partial \mathcal{L}}{\partial \dot{q}_l} &= l_l^2\frac{m_l}{q_l^2}\dot{q}_l + 2l_l^2 \frac{m_l}{q_l^4}\dot{q}_l - 2l_l^2m_l\frac{\cos{q_l}}{q_l^4}\dot{q}_l  - 2l_l^2m_l \frac{\sin{q_l}}{q_l^3}\dot{q}_l - l_l\frac{m_l}{q_l^2}\dot{z}_v - l_lm_l\frac{\cos{q_l}}{q_l}\dot{x}_v   \\
&\hspace{10mm} + l_lm_l\frac{\cos{q_l}}{q_l^2}\dot{z}_v + l_lm_l\frac{\sin{q_l}}{q_l^2}\dot{x}_v  + l_lm_l\frac{\sin{q_l}}{q_l}\dot{z}_v  + I_l\dot{q}_l s \\
\frac{\partial \mathcal{L}}{\partial \dot{q}_r} &= l_r^2\frac{m_r}{q_r^2}\dot{q}_r + 2l_r^2\frac{m_r}{q_r^4}\dot{q}_r - 2l_r^2m_r\frac{\cos{q_r}}{q_r^4}\dot{q}_r  - 2l_r^2m_r\frac{\sin{q_r}}{q_r^3}\dot{q}_r - l_r\frac{m_r}{q_r^2}\dot{z}_v + l_rm_r\frac{\cos{q_r}}{q_r}\dot{x}_v  \\
&\hspace{10mm}+ l_rm_r\frac{\cos{q_r}}{q_r^2}\dot{z}_v - l_rm_r\frac{\sin{q_r}}{q_r^2}\dot{x}_v + l_rm_r\frac{\sin{q_r}}{q_r}\dot{z}_v + I_r\dot{q}_r .
\end{aligned}
\end{equation}

Now, we compute:
\begin{equation}\label{eq:ddtdldqdot}
    \frac{d}{dt} \left(\frac{\partial \mathcal{L}}{\partial \dot{\textbf{q}}} \right) = \begin{pmatrix}
        a_1 & a_2 & a_3 & a_4 & a_5
    \end{pmatrix}^{\intercal}
\end{equation}
where
\begin{equation}
\begin{aligned}
        a_1 &= (m+m_l)\ddot{x}_v  -  l_lm_l\left(  \frac{\cos{q_l}}{q_l} \right) \ddot{q_l}  + l_lm_l \left( \frac{q_l\sin{q_l} + \cos{q_l}}{q_l^2} \right) \dot{q_l}^2 + l_lm_l\left(  \frac{\sin{q_l}}{q_l^2} \right) \ddot{q_l}  \\
        &\hspace{5mm} + l_lm_l \left( \frac{q_l\cos{q_l} - 2\sin{q_l}}{q_l^3} \right) \dot{q_l}^2 + m_r\ddot{x}_v   + l_rm_r\left(  \frac{\cos{q_r}}{q_r} \right) \ddot{q_r} - l_rm_r \left( \frac{q_r\sin{q_r} + \cos{q_r}}{q_r^2} \right) \dot{q_r}^2 \\
        &\hspace{10mm}- l_rm_r\left(  \frac{\sin{q_r}}{q_r^2} \right) \ddot{q_r} - l_rm_r \left( \frac{q_r\cos{q_r} - 2\sin{q_r}}{q_r^3} \right) \dot{q_r}^2 \\
        a_2 &= (m+m_l)\ddot{z}_v -l_lm_l\left(\frac{1}{q_l^2}\right)\ddot{q_l} + 2l_lm_l \left(\frac{1}{q_l^3}\right)\dot{q_l}^2  + l_lm_l\left(  \frac{\cos{q_l}}{q_l^2} \right) \ddot{q}_l - l_lm_l \left( \frac{q_l\sin{q_l} + 2\cos{q_l}}{q_l^3} \right) \dot{q_l}^2 \\
        &\hspace{5mm} + l_lm_l\left(  \frac{\sin{q_l}}{q_l} \right) \ddot{q}_l + l_lm_l \left( \frac{q_l\cos{q_l} - \sin{q_l}}{q_l^2} \right) \dot{q}_l^2 + m_r\ddot{z}_v -l_rm_r \left(\frac{1}{q_r^2}\right)\ddot{q}_r + 2l_rm_r \left(\frac{1}{q_r^3}\right)\dot{q}_r^2 \\
         &\hspace{10mm} + l_rm_r\left(  \frac{\cos{q_r}}{q_r^2} \right) \ddot{q}_r - l_rm_r \left( \frac{q_r\sin{q_r} + 2\cos{q_r}}{q_r^3} \right) \dot{q}_r^2 + l_rm_r\left(  \frac{\sin{q_r}}{q_r} \right) \ddot{q}_r + l_rm_r \left( \frac{q_r\cos{q_r} - \sin{q_r}}{q_r^2} \right) \dot{q}_r^2 \\
         a_3 &= I\ddot{\theta} \\
        a_4 &= l_l^2m_l\left(\frac{1}{q_l^2}\right)\ddot{q}_l - 2l_l^2m_l \left(\frac{1}{q_l^3}\right)\dot{q}_l^2 + 2l_l^2m_l\left(\frac{1}{q_l^4}\right)\ddot{q}_l  - 8l_l^2m_l \left(\frac{1}{q_l^5}\right)\dot{q}_l^2
        -2l_l^2m_l\left(  \frac{\cos{q_l}}{q_l^4} \right) \ddot{q}_l \\
        &\hspace{5mm} + 2l_l^2m_l \left( \frac{q_l\sin{q_l} + 4\cos{q_l}}{q_l^5} \right) \dot{q}_l^2 -2l_l^2m_l\left(  \frac{\sin{q_l}}{q_l^3} \right) \ddot{q}_l - 2l_l^2m_l \left( \frac{q_l\cos{q_l} - 3\sin{q_l}}{q_l^4} \right) \dot{q}_l^2 
        -l_lm_l \left(\frac{1}{q_l^2}\right)\ddot{z}_v  \\
        &\hspace{10mm}   +  2l_lm_l\left(\frac{1}{q_l^3}\right)\dot{q}_l\dot{z}_v - l_lm_l\left(  \frac{\cos{q_l}}{q_l} \right) \ddot{x}_v + l_lm_l \left( \frac{q_l\sin{q_l} + \cos{q_l}}{q_l^2} \right) \dot{q_l}\dot{x}_v
        + l_lm_l\left(  \frac{\cos{q_l}}{q_l^2} \right) \ddot{z}_v \\
        &\hspace{15mm}    - l_lm_l \left( \frac{q_l\sin{q_l} + 2\cos{q_l}}{q_l^3} \right) \dot{q}_l\dot{z}_v + l_lm_l\left(  \frac{\sin{q_l}}{q_l^2} \right) \ddot{x}_v  + l_lm_l \left( \frac{q_l\cos{q_l} - 2\sin{q_l}}{q_l^3} \right) \dot{q_l}\dot{x}_v \\
        &\hspace{20mm}  
        + l_lm_l\left(  \frac{\sin{q_l}}{q_l} \right) \ddot{z}_v  + l_lm_l \left( \frac{q_l\cos{q_l} - \sin{q_l}}{q_l^2} \right) \dot{q_l}\dot{z}_v + I_l \ddot{q}_l
\end{aligned}
\end{equation}
and
\begin{equation}
\begin{aligned}
a_5 &= l_r^2m_r\left(\frac{1}{q_r^2}\right)\ddot{q}_r - 2l_r^2m_r \left(\frac{1}{q_r^3}\right)\dot{q}_r^2
+ 2l_r^2m_r\left(\frac{1}{q_r^4}\right)\ddot{q}_r - 8l_r^2m_r \left(\frac{1}{q_r^5}\right)\dot{q}_r^2
-2l_r^2m_r\left(  \frac{\cos{q_r}}{q_r^4} \right) \ddot{q}_r \\
&\hspace{5mm} + 2l_r^2m_r \left( \frac{q_r\sin{q_r} + 4\cos{q_r}}{q_r^5} \right) \dot{q}_r^2 - 2l_r^2m_r\left(  \frac{\sin{q_r}}{q_r^3} \right) \ddot{q}_r   - 2l_r^2m_r \left( \frac{q_r\cos{q_r} - 3\sin{q_r}}{q_r^4} \right) \dot{q}_r^2
-l_rm_r\left(\frac{1}{q_r^2}\right)\ddot{z}_v \\
&\hspace{10mm}+  2l_rm_r\left(\frac{1}{q_r^3}\right)\dot{q}_r\dot{z}_v  + l_rm_r\left(  \frac{\cos{q_r}}{q_r} \right) \ddot{x}_v  - l_rm_r \left( \frac{q_r\sin{q_r} + \cos{q_r}}{q_r^2} \right) \dot{q}_r\dot{x}_v
+ l_rm_r\left(  \frac{\cos{q_r}}{q_r^2} \right) \ddot{z}_v  \\
&\hspace{15mm} - l_rm_r \left( \frac{q_r\sin{q_r} + 2\cos{q_r}}{q_r^3} \right) \dot{q}_r\dot{z}_v - l_rm_r\left(  \frac{\sin{q_r}}{q_r^2} \right) \ddot{x}_v   - l_rm_r \left( \frac{q_r\cos{q_r} - 2\sin{q_r}}{q_r^3} \right) \dot{q}_r\dot{x}_v  \\
&\hspace{25mm} + l_rm_r\left(  \frac{\sin{q_r}}{q_r} \right) \ddot{z}_v + l_rm_r \left( \frac{q_r\cos{q_r} - \sin{q_r}}{q_r^2} \right) \dot{q}_r\dot{z}_v + I_r \ddot{q}_r
\end{aligned}
\end{equation}

Now, we proceed to compute:
\begin{equation}
\begin{aligned}
 \frac{\partial \mathcal{L}}{\partial \textbf{q}} &= \begin{pmatrix}
        \frac{\partial \mathcal{L}}{\partial {x}_v} &
        \frac{\partial \mathcal{L}}{\partial {z}_v} &
        \frac{\partial \mathcal{L}}{\partial {\theta}} &
        \frac{\partial \mathcal{L}}{\partial {q}_l} &
        \frac{\partial \mathcal{L}}{\partial {q}_r} 
    \end{pmatrix}^{\intercal}\\
&=\begin{pmatrix}
0 & -g(m + m_l + m_r) & 0 & \frac{\partial \mathcal{L}}{\partial {q}_l} & \frac{\partial \mathcal{L}}{\partial {q}_r} 
 \end{pmatrix}^{\intercal}
\end{aligned}
\end{equation}
where 
\begin{equation}
    \begin{aligned}
\frac{\partial \mathcal{L}}{\partial {q}_l}&=  - l_{l}^{2}m_{l} \frac{1}{q_{l}^{3}}\dot{q}_l^{2} - 4l_{l}^{2}m_{l}\frac{1}{q_{l}^{5}}
\dot{q}_l^{2}  + l_{l}^{2}m_{l} \left( \frac{q_l\sin{q_l} + 4\cos{q_l}}{q_l^5} \right)\dot{q}_l^{2}  - l_{l}^{2}m_{l} \left( \frac{q_l\cos{q_l} - 3\sin{q_l}}{q_l^4} \right)\dot{q}_l^{2} + 2l_{l}m_{l}
\frac{1}{q_{l}^{3}}\dot{q}_l\dot{z}_v \\
&\hspace{5mm} + l_{l}m_{l} \left( \frac{q_l\sin{q_l} + \cos{q_l}}{q_l^2} \right)\dot{q}_l\dot{x}_v  + l_{l}
m_{l} \left( \frac{q_l\cos q_{l} - 2\sin{q_l}}{q_{l}^{3}} \right)\dot{q}_l\dot{x}_v  - l_{l}m_{l} \left( \frac{ q_l\sin{q_l} + 2\cos{q_l}}{q_{l}^{3}} \right)\dot{q}_l\dot{z}_v \\
&\hspace{10mm} + l_{l}m_{l} \left( \frac{q_l\cos q_{l} - \sin{q_l}}{q_{l}^2} \right)\dot{q}_l\dot{z}_v - (g m_l l_l)\left(  \frac{q_l\sin{q_l} + \cos{q_l} - 1}{q_l^2} \right)
    \end{aligned}
\end{equation}
and  
\begin{equation}
    \begin{aligned}
        \frac{\partial \mathcal{L}}{\partial {q}_r} &= - l_{r}^{2} m_{r}  \frac{1}{q_{r}^{3}} \dot{q}_r^{2}
         - 4l_{r}^{2}m_{r}  \frac{1}{q_{r}^{5}} \dot{q}_r^{2}   + l_{r}^{2}m_{r} \left( \frac{q_r\sin{q_r} + 4\cos{q_r}}{q_{r}^{5}} \right) \dot{q}_r^{2}  - l_{r}^{2}m_{r} \left( 
 \frac{q_r\cos{q_r} - 3\sin{q_r}}{q_r^{4}} \right)\dot{q}_r^{2}
 + 2l_{r}m_{r}
\frac{1}{q_{r}^{3}}\dot{q}_r\dot{z}_v \\
        &\hspace{5mm} - l_{r}m_{r} \left( \frac{q_r\sin{q_r} + \cos{q_r}}{q_r^2} \right)\dot{q}_r\dot{x}_v - l_{r}
m_{r} \left( \frac{q_r\cos{q_r} -2\sin{q_r}}{q_{r}^{3}} \right)\dot{q}_r\dot{x}_v - l_{r}m_{r} \left( \frac{q_r\sin{q_r} + 2\cos{q_r}}{q_{r}^{3}} \right)\dot{q}_r\dot{z}_v \\
&\hspace{10mm} + l_{r}m_{r} \left( \frac{q_r\cos{q_r} - \sin{q_r}}{q_{r}^2} \right)\dot{q}_r\dot{z}_v - (gm_r l_r)\left(  \frac{q_r\sin{q_r} + \cos{q_r} - 1}{q_r^2}\right)
    \end{aligned}
\end{equation}

\section{Proof of Lemma \ref{lemm:matrix}}\label{app:lemma_matrix}
To prove the positiveness of $D(\mathbf{q})$, it suffices to demonstrate that the principal minors of $D$ are always positive. It is clear by the definition of parameters that the first three principal minors are positive. The fourth principal minor is computed as follows:
\begin{equation}
\textrm{det}(D_4)=
\Theta _{3} 
\begin{vmatrix}
\rho _{1} & 0 & 0 & \mathscr{D}_{1} \\ 
0 & \rho _{1} & 0 & \mathscr{D}_{2} \\ 
0 & 0 & \rho _{2} & 0 \\ 
\mathscr{D}_{1} & \mathscr{D}_{2} & 0 & l_l\mathscr{D}_{5}+\rho_{3}%
\end{vmatrix}
\end{equation}
where $\rho _{1}=\frac{\Theta _{1}}{\Theta _{3}},$ $\rho _{2}=\frac{\Theta _{2}}{\Theta _{3}},$ $\rho _{3}=\frac{\Theta _{5}}{\Theta _{3}}$. After several computations, it follows that:
\begin{equation}
    \begin{aligned}
    \textrm{det}(D_4) &= -\frac{\Theta_3}{q_l^{4}}\Bigl( -\rho_{2}\rho _{3}q_l^{4}\rho_{1}^{2} - l_l \rho_{2}q_l^{2}\rho_{1}^{2} + \rho_{2}q_l^{2}\rho_{1}\cos^{2}{q_l}  +  
 \rho_{2}q_l^{2}\rho_{1}\sin^{2}{q_l} + 2l_l\rho_{2}q_l\rho_{1}^{2}\sin{q_l} - 2\rho_{2}l\rho_{1}\sin{q_l} \\ 
 &\hspace{10mm} + 2l_l\rho_{2}\rho_{1}^{2}\cos{q_l} - 2l_l\rho_{2}\rho_{1}^{2} + \rho_{2}\rho_{1}\cos^{2}{q_l}  - 2\rho_{2}\rho_{1}\cos{q_l} + \rho_{2}\rho_{1}\sin^{2}{q_l} + \rho_{2}\rho_{1}\Bigr) \\
&= \frac{\Theta_3}{q_l^{4}}\rho_{1}\rho _{2}\Bigl( 2\cos{q_l} + 2 l_l \rho_{1} + 2q_l\sin{q_l} -q_l^{2}  - 2l_l\rho_{1}\cos{q_l}
+ l_l q_l^{2}\rho_{1}+q_l^{4}\rho_{1}\rho_{3}-2 l_l q_l\rho_{1}\sin{q_l} -2\Bigr) \\
&= \frac{\Theta_3}{q_l^{4}}\rho_{1}\rho_{2}\Bigl[\Bigl( -q_l^{2}-2+2\cos{q_l}+2q_l\sin{q_l} \Bigr)  +  l_l \rho_{1} \Bigl( q_l^{2}  + 2 -2  \cos{q_l} - 2q_l \sin{q_l}  \Bigr) +\rho_{1}\rho_{3} q_l^{4}   \Bigr] \\
&= -\left( \Theta_3 \rho_{1}\rho_{2}\right) \left[ \underbrace{\left( \frac{q_l^{2}+2-2\cos{q_l}
- 2 q_l\sin{q_l} }{q_l^{4}}\right)}_{\xi_1} + \left( 2 l_l \rho_{1}\right) \underbrace{\left( \frac{\cos{q_l} + q_l\sin{q_l} -1-\frac{1}{2}q_l^{2}-q_l^{4}\frac{\rho_{3}}{2 l_l }}{q_l^{4}}\right)}_{\xi_2} \right].
    \end{aligned}
\end{equation}
Simple computations show that,
\begin{equation}
    \lim_{q_l\to 0 } \xi_1 = \frac{1}{4}, \quad \lim_{q_l\to 0 } \xi_2 =  -\frac{1}{8}-\frac{\rho _{3}}{2 l_l}
\end{equation}
And also that $\lim_{q_l\to \infty } \xi_1 = \lim_{q_l\to -\infty } \xi_1 = 0$ and $ \lim_{q_l\to \infty } \xi_2 = \lim_{q_l\to -\infty } \xi_2 = -\frac{\rho_3}{2l_l}$. Thus, one can demonstrate that, at least locally, the maximum of $\xi_1$ and the minimum of $\xi_2$ are in $q_l=0$. From the above, the worst-case scenario occurs when $q_l=0$, and then:
\begin{equation}
\Theta_3 \rho_{1}\rho_{2} \left[ \left( 2 l_l \rho_{1}\right) \left(\frac{1}{8} + \frac{\rho_{3}}{2 l_l }\right) - \frac{1}{4}\right] \leq \textrm{det}(D_4) \leq \Theta_3 \rho_{1}\rho_{2}
\end{equation}
as long as
\begin{equation}
    \begin{aligned}
        \frac{1}{4} < \left( 2 l_l \rho _{1}\right) \left( \frac{1}{8} + \frac{\rho_{3}}{2l_l}\right) < 1%
    \end{aligned}
\end{equation}
or from the LHS of the last inequality, $\rho _{1}>\frac{1}{l_l+4\rho _{3}}$, equivalently $\Theta _{1}>\frac{\Theta _{3}^{2}}{l_l\Theta _{3}+4\Theta _{5}}$.

Now, we compute $\textrm{det}(D(\textbf{q}))$. Let us assume for a moment that $\rho_4 = \frac{\Theta_4}{\Theta_3}$, $\rho_5 = \frac{\Theta_6}{\Theta_3}$:
\begin{equation}
    \textrm{det}(D) = \Theta_3\begin{vmatrix}
        \rho_1 & 0 & 0 &    \mathscr{D}_{1}   & \rho_4 \mathscr{D}_{3} \\
        0 & \rho_1 & 0 & \mathscr{D}_{2} &    \rho_4 \mathscr{D}_{4}   \\
        0 & 0 & \rho_2 & 0 & 0 \\
          \mathscr{D}_{1} & \mathscr{D}_{2} & 0 & l_l \mathscr{D}_{5} + \rho_3  & 0 \\
          \rho_4 \mathscr{D}_{3}   & \rho_4 \mathscr{D}_{4} & 0 & 0 & l_r\rho_4 \mathscr{D}_{6} + \rho_5
    \end{vmatrix}
\end{equation}
where $\rho_4 = \frac{\Theta_6}{\Theta_3}$. Thus\footnote{This is computed with Mathematica software.}, 
\begin{equation}
    \begin{aligned}
\textrm{det}(D(\textbf{q})) &=\frac{\Theta_3\rho_2}{q_l^4 q_r^4}\Biggl( \rho_1 \Bigl[ -2+2 l_l \rho_1-q_l^2+l_l \rho_1 q_l^2+\rho_1 
\rho_3 q_l^4 + (2-2 l_l \rho_1) \cos{q_l}-2 (-1+l_l \rho_1) q_l \sin{q_l} \Bigr] \Bigl[ 2 \rho_4 l_r  \\
&\hspace{5mm} + \rho_4 l_r q_r^2 + \rho_5 q_r^4-2 \rho_4 l_r \cos{q_r}-2 \rho_4 l_r q_r \sin{q_r} \Bigr] + \rho_4^2 \Bigl( -1+\cos{q_r}+q_r \sin{q_r} \Bigr) \Bigl(  -\rho_1 \Bigl[ 2 l_l+l_l q_l^2+
\rho_3 q_l^4 \\
&\hspace{10mm}- 2 l_l \cos{q_l} - 2 l_l q_l \sin{q_l} \Bigr]  \Bigl[ -1+\cos{q_r}+q_r \sin{q_r} \Bigr]  + \Bigl[ q_l \cos{q_l}-\sin{q_l} \Bigr] \Bigl[ - q_l \cos{q_l} - q_r \cos{q_r}  \\
&\hspace{15mm} + q_l \cos{(q_l+q_r)} + q_r \cos{(q_l+q_r)} + \sin{q_l} + \sin{q_r} - \sin{(q_l+q_r)} + q_l q_r \sin{(q_l+q_r)} \Bigr] \Bigr)\\
&\hspace{20mm}  - \rho_4^2 \Bigl( q_r \cos{q_r}-\sin{q_r} \Bigr) \Bigl( \rho_1 \Bigl[ 2 l_l+l_l q_l^2  + \rho_3 q_l^4-2 l_l \cos{q_l}-2 l_l q_l \sin{q_l} \Bigr] \Bigl[ q_r \cos{q_r}-\sin{q_r} \Bigr] \\
&\hspace{25mm} - \Bigl[ -1 +\cos{q_l}+q_l \sin{q_l} \Bigr] \Bigr[ -q_l \cos{q_l}-q_r \cos{q_r} +q_l \cos{(q_l+q_r)} \\
&\hspace{30mm}+q_r \cos{(q_l+q_r)}+\sin{q_l}+\sin{q_r}-\sin{(q_l+q_r)} +q_l q_r \sin{(q_l+q_r)} \Bigl] \Bigr) \Biggr) .
    \end{aligned}
\end{equation}
Proceeding as before, computing the limits one gets:
\begin{equation}
\begin{aligned}
    \lim_{q_l,q_r\to 0 } \textrm{det}(D(\textbf{q}))  &=  
    \left( \frac{\Theta_3 \rho_1 \rho_2 }{16} \right) \Bigl(- \rho_4 l_r -4 \rho_5   -4 \rho_3 \Bigl[ \rho_4^2 - \rho_1 \rho_4 l_r -4 \rho_1 \rho_5 \Bigr] + l_l \Bigl[ -\rho_4^2 + \rho_1 \rho_4 l_r +4 \rho_1 \rho_5 \Bigr] \Bigr) \\
    &= \left( \frac{\Theta_3 \rho_1 \rho_2 }{16} \right) \Bigl(- \rho_4 l_r - 4 \rho_5 - 4 \rho_3  \rho_4^2 + 4 \rho_3 \rho_1 \rho_4 l_r  + 16 \rho_3 \rho_1 \rho_5  - l_l \rho_4^2 + l_l \rho_1 \rho_4 l_r +4  \rho_1 \rho_5 l_l \Bigr) \\
    &= \left( \frac{\Theta_3 \rho_1 \rho_2 }{16} \right) \Bigl( \Bigl[ - \rho_4 l_r - 4 \rho_5 - 4 \rho_3  \rho_4^2 - l_l \rho_4^2 \Bigr]   + \Bigl[  4 \rho_3 \rho_1 \rho_4 l_r  + 16 \rho_3 \rho_1 \rho_5 + l_l \rho_1 \rho_4 l_r +4  \rho_1 \rho_5 l_l 
       \Bigr]   \Bigr) ,
\end{aligned}
\end{equation}
and $\textrm{det}(D(\textbf{q})) > 0$ as long as,
\begin{equation}
    4 \rho_3 \rho_1 \rho_4 l_r  + 16 \rho_3 \rho_1 \rho_5 + l_l \rho_1 \rho_4 l_r +4  \rho_1 \rho_5 l_l  >   \rho_4 l_r + 4 \rho_5 + 4 \rho_3  \rho_4^2 + l_l \rho_4^2.
\end{equation}
And also:
\begin{equation}
\lim_{q_l,q_r\to \infty } \textrm{det}(D(\textbf{q})) = \lim_{q_l,q_r\to -\infty } \textrm{det}(D(\textbf{q}))  =  \Theta_3\rho_1^2 \rho_2 \rho_3 \rho_5.
\end{equation}
Hence, given that all the principal minors of $D(\mathbf{q})$ are positive, we conclude that the matrix $D(\mathbf{q})$ is positive definite. This indicates that $D(\mathbf{q})$ is symmetric and possesses all positive eigenvalues, contributing to its positive definite nature.
\section{Proof of Lemma \ref{eq:SkewMatrix}}\label{app:lemmaSkew}
Now, we compute $\dot{D}(\textbf{q})$ to verify that $\dot{D}(\textbf{q}) - 2C(\textbf{q}, \dot{\textbf{q}})$ is skew-symmetric. Notice that,
\begin{equation}\label{eq:inertia_dot}
    \dot{D} = \begin{pmatrix}
        0 & 0 & 0 & \Theta_3 \mathscr{C}_{1} \dot{q}_l & - \Theta_4 \mathscr{C}_{3} \dot{q}_r   \\
        0 & 0 & 0 & \Theta_3 \mathscr{C}_{2} \dot{q}_l & \Theta_4 \mathscr{C}_{4}  \dot{q}_r  \\
        0 & 0 & 0 & 0 & 0 \\
        \Theta_3 \mathscr{C}_{1} \dot{q}_l & \Theta_3 \mathscr{C}_{2} \dot{q}_l & 0 & -4l_l\Theta_3 \mathscr{C}_{5} \dot{q}_l  & 0 \\
         - \Theta_4 \mathscr{C}_{3} \dot{q}_r    & \Theta_4 \mathscr{C}_{4} \dot{q}_r   & 0 & 0 & -4l_r\Theta_4 \mathscr{C}_{6} \dot{q}_r
    \end{pmatrix},
\end{equation}
and 
\begin{equation}\label{eq:Coriolis_reduced_subs}
    -2C(\textbf{q}, \dot{\textbf{q}}) = \begin{pmatrix}
        0 & 0 & 0 & -2\Theta_3 \mathscr{C}_{1} \dot{q_l} & 2 \Theta_4 \mathscr{C}_{3}  \dot{q_r} \\
        0 & 0 & 0 &  -2\Theta_3 \mathscr{C}_{2} \dot{q}_l &  -2\Theta_4 \mathscr{C}_{4} \dot{q}_r \\
        0 & 0 & 0 & 0 & 0 \\
        0 & 0 & 0 & 4l_l\Theta_3 \mathscr{C}_{5} \dot{q}_l & 0  \\
        0 & 0 & 0 & 0 & 4l_r\Theta_4 \mathscr{C}_{6} \dot{q}_r 
    \end{pmatrix},
\end{equation}
and finally, we compute
\begin{equation}\label{eq:skew}
\begin{aligned}
    \dot{D}(\textbf{q}) - 2C(\textbf{q}, \dot{\textbf{q}}) &=  \begin{pmatrix}
        0 & 0 & 0 & -\Theta_3 \mathscr{C}_{1} \dot{q}_l &  \Theta_4 \mathscr{C}_{3} \dot{q}_r   \\
        0 & 0 & 0 & -\Theta_3 \mathscr{C}_{2} \dot{q}_l    & -\Theta_4 \mathscr{C}_{4} \dot{q}_r  \\
        0 & 0 & 0 & 0 & 0 \\
        \Theta_3 \mathscr{C}_{1} \dot{q}_l   & \Theta_3 \mathscr{C}_{2} \dot{q}_l & 0 & 0  & 0 \\
         - \Theta_4 \mathscr{C}_{3} \dot{q}_r    & \Theta_4 \mathscr{C}_{4} \dot{q}_r   & 0 & 0 & 0
    \end{pmatrix},
\end{aligned}
\end{equation}
which clearly is skew-symmetric.

\begin{multicols}{2}


\bibliographystyle{plain}        
\bibliography{sample}            
\end{multicols}

\end{document}